\documentclass[11pt]{article}
\usepackage{a4wide}
\usepackage{color,epsfig}
\usepackage{amsmath}
\usepackage{amsfonts}
\usepackage{dsfont}
\usepackage{mathrsfs}
\usepackage{amssymb}
\usepackage{amsfonts,graphicx,psfrag}
\usepackage{cite}
\usepackage{amsthm}
\usepackage{hyperref}
\usepackage{enumerate}
\usepackage{mathtools}
\usepackage{enumitem}
\usepackage{xcolor}
\usepackage{tikz}

\def\R{{\rm R}}

\def\diag{{\rm diag}}

\def\<{\langle}
\def\>{\rangle}
\def\lb{\label}

\def\pt{\partial}

\newtheorem{lemma}{Lemma}[section]
\newtheorem{theorem}[lemma]{Theorem}
\newtheorem{corollary}[lemma]{Corollary}

\newtheorem{proposition}[lemma]{Proposition}
\newtheorem{definition}[lemma]{Definition}

\newtheorem{remark}[lemma]{Remark}
\newtheorem{question}{Question}[section]

\title{Static equilibrium of the charged $N$-body problem}
\author{Xuhui Hu$^{1}$\thanks{E-mail: 22335014@zju.edu.cn.},\quad and \quad
	Qinglong Zhou$^{1,2}$\thanks{Partially supported by National Natural Science Foundation of China (No.12171426), the Natural Science Foundation of Zhejiang Province (No. LY19A010020) and the Fundamental Research Funds for the Central Universities (No. 2024FZZX02-01-01).
		E-mail: zhouqinglong@zju.edu.cn. }\\	
	$^{1}$ Department of Mathematics,\\Zhejiang University, Hangzhou 310058, Zhejiang, China\\
	$^{2}$ Institut f\"ur Mathematics,\\Universit\"at Augsburg, 86159 Augsburg, Germany\\
}

\begin{document}

\maketitle

\begin{abstract}
We study non-collision static equilibria in the charged \(N\)-body problem.
Unlike the classical Newtonian case, such equilibria may exist because the
effective interaction coefficients $\delta_{ij}=m_im_j-e_ie_j$
can be positive, negative, or zero. We give a complete classification of
three-body equilibria and prove that all of them are collinear and linearly
unstable. For the four-body problem, we derive sign restrictions for convex
and concave non-collinear equilibria, exclude non-trivial concyclic
quadrilateral equilibria, and obtain a geometric necessary condition expressed
through an auxiliary triangle and a resultant equation. We also study collinear
equilibria by reducing the force-balance equations to a homogeneous polynomial
system. This yields a resultant necessary condition and an inverse realization
theorem showing that every prescribed collinear configuration of distinct
points can be realized by suitable positive masses and real charges. 
Finally, we analyze regular polygon
configurations and centered regular polygon configurations, showing in
particular that equal masses cannot form a non-trivial regular \(N\)-gon
equilibrium, while the centered case reduces to a single scalar condition.
\end{abstract}

\renewcommand{\theequation}{\thesection.\arabic{equation}}
\section{Introduction}

\subsection{Background and motivation}

{In the charged planar $N$-body problems of celestial mechanics, the position vectors of the $N$-particles are denoted by $q_1 ,\dots , {q_N}\in \R^2$, and the masses are represented by
$m_1 ,\dots,{m_N} > 0$.} By Newton’s second law and the law of universal
gravitation, the system of equations is
\begin{align}
m_i\ddot{q_i} = \frac{\pt U}{\pt q_i}, \quad i = 1, {\dots, N},\label{1.1}
\end{align}
where $U(q) = U(q_1, {\dots, q_{N}}) = \sum_{1\leq i< j\leq {N}} \frac{m_i m_j-e_ie_j}{|q_i-q_j|}$ 
is the potential function and $|\cdot |$ is the standard norm of a vector in $\R^2$. For convenience, we define
$$
\delta_{ij}=m_im_j - e_ie_j
$$
and
$$
\lambda_{ij}=1-\frac{e_i}{m_i}\frac{e_j}{m_j}
$$
for $1\le i,j\le N$ and $i\ne j$.
The case \(\delta_{ij}>0\) (or equivalently, \(\lambda_{ij}>0\)) corresponds to an attractive effective interaction,
whereas \(\delta_{ij}<0\) (or equivalently, \(\lambda_{ij}<0\)) corresponds to a repulsive effective interaction.

This problem is called the charged $N$-body problem and can be considered as a natural
generalization of the gravitational $N$-body problem. The charged two-body problem is integrable for an arbitrary value of $\lambda_{12}$, see, e.g., \cite{BeV}, Sec 2.

For $N>2$, if at most one of the values $\lambda_{ij},1\le i<j\le N$ is nonzero, such a charged $N$-body problem can be transformed into a charged two-body problem and is therefore integrable.
Consequently, we always make the following assumption throughout this paper.

\medskip
{\bf Non-trivial Assumption}: {\it At least two of the values $\{\lambda_{ij}\;|\;1\le i<j\le N\}$ are nonzero.}
\medskip

It is {well-known} that \eqref{1.1} can be reformulated as a Hamiltonian system by 
\begin{align}
\dot{p}_i = -\frac{\pt H}{\pt q_i}, \; \dot{q}_i  = \frac{\pt H}{\pt p_i}, \quad \mbox{for} \;i = 1, {\dots, N}, \label{1.2}
\end{align}
with the Hamiltonian function
\begin{align} 
H(p, q) = \sum_{i =1}^{{N}} \frac{|p_i|^2}{2m_i} - U(q_1,{\dots, q_{N}}).\lb{1.3}
\end{align}


Finding solutions to the equations \eqref{1.1} is in general very difficult. 
An important building block for such solutions are central configurations which are special points in configuration space that have the property that for each body, the force is proportional to the distance vector of the body to the common center of masses, where the proportionality constant is the same for all bodies.

\begin{definition}\label{def:cc}
    A configuration $a = (a_1, a_2,\ldots, a_N)\in(\mathbb{R}^k)^N$ with $a_i\ne a_j$ , $\forall1\le i<j\le N$ is
a central configuration for the given mass $m = (m_1, m_2,\ldots, m_N)\in(\mathbb{R}^+)^N$ and the quantities of
charges $e = (e_1,e_2,\ldots, e_N) \in\mathbb{R}^N$, if there exists some $\lambda\in\mathbb{R}$ such that $(a, \lambda)$ is a solution of the
algebraic system
\begin{equation}\label{eq.of.cc}
    \lambda Ma+\frac{\partial U(a)}{\partial q}=0,
\end{equation}
with $M = \diag(m_1I_k,\ldots, m_nI_k)$. By the homogeneity of U of degree $-1$, \eqref{eq.of.cc} implies
$$
\lambda={U(a)}/{(Ma\cdot a)}.
$$
\end{definition}

\begin{definition}\label{def:e}
    A configuration $a = (a_1, a_2,\ldots, a_N)\in(\mathbb{R}^k)^N$ with $a_i\ne a_j$ , $\forall1\le i<j\le N$ is
an equilibrium for the given mass $m = (m_1, m_2,\ldots, m_N)\in(\mathbb{R}^+)^N$ and the quantities of
charges $e = (e_1,e_2,\ldots, e_N) \in\mathbb{R}^N$, 
if $a = (a_1, a_2,\ldots, a_N)$ is a solution of the
algebraic system
\begin{equation}\label{eq.of.e}
    \frac{\partial U(a)}{\partial q}=0.
\end{equation}
By the homogeneity of U of degree $-1$, \eqref{eq.of.e} implies
$$
U(a)=0.
$$
\end{definition}

By the above two definitions, an equilibrium is a special case of central configuration.
In the gravitational $N$-body problem, equilibrium states never exist due to $\lambda={U(a)}/{(Ma\cdot a)} > 0$. However, in the charged $N$-body problem, equilibrium states can potentially exist. 

Central configurations in charged problems have been studied previously in
several settings. To mention only a few representative works, we refer to
studies of the charged three-body problem, the charged isosceles three-body
problem, continua of central configurations in charged problems, linear
stability of charged three-body relative equilibria, and four-body
configurations with gravitational charges of both signs; see
\cite{Atela1988,PSSY96,AlP2002,AlP2008,PiL2012, ZhL2015}.
However, the static equilibria considered in this paper, defined by
\(\nabla U=0\), appear to have been much less systematically studied.
Our main objective is to investigate the conditions for the existence of 
non-collision static equilibria and their properties in the charged N-body problem.


The present work is also motivated by the classical \(N\)-center problem.
In the planar \(N\)-center problem, one studies the motion of a test particle
in the gravitational field generated by \(N\) fixed primary centers. More
precisely, if the fixed centers are located at \(c_1,\ldots,c_N\in\mathbb C\)
with positive masses \(m_1,\ldots,m_N\), and if
\[
    X=\mathbb C\setminus\{c_1,\ldots,c_N\},
\]
then the position \(x(t)\in X\) of the test particle satisfies an equation of the form
\[
    \ddot x(t)
    =
    -\sum_{j=1}^{N}
    \frac{m_j}{|x(t)-c_j|^{\alpha+2}}
    \bigl(x(t)-c_j\bigr),
    \qquad \alpha>0.
\]
The Newtonian case corresponds to \(\alpha=1\). The \(N\)-center problem is a
classical model in celestial mechanics and singular Hamiltonian systems. It has
been studied from several viewpoints, including variational methods for singular
potentials, collision avoidance, topological constraints, symbolic dynamics, and
rotating \(N\)-center problems; see, for example,
\cite{Gordon1975,Gordon1977,SoaveTerracini2012,Soave2014,Yu2016NCenter}
and references therein.

However, from the viewpoint of classical celestial mechanics, the assumption
that the \(N\) primary bodies remain fixed is an idealization rather than a
self-consistent Newtonian \(N\)-body motion. Indeed, in the purely gravitational
\(N\)-body problem, all mutual interactions are attractive, and hence a
non-collision static equilibrium of the \(N\) massive bodies cannot occur. In
other words, the fixed centers in the classical \(N\)-center problem have to be
regarded as prescribed external sources, rather than as bodies forming an
isolated static celestial system.

The charged \(N\)-body problem provides a natural way to remove this difficulty.
Since each pairwise interaction has the effective coefficient
$\delta_{ij}=m_i m_j-e_i e_j$,
the gravitational attraction and the electrostatic interaction may balance one
another. Therefore, for suitable masses and charges, the \(N\) primary bodies may
form a genuine static equilibrium. Once such an equilibrium configuration is
obtained, one may consider a massless and uncharged test particle moving in the
field generated by these \(N\) fixed charged bodies. In this sense, the
\(N\)-center problem can be interpreted as the restricted problem associated with
an equilibrium configuration of the charged \(N\)-body problem. This observation
gives an additional motivation for studying the existence and structure of
equilibria in the charged \(N\)-body problem.

\subsection{Main results}

The main purpose of this paper is to study the existence and qualitative
properties of equilibria in the charged \(N\)-body problem. We summarize our
main results below.

Our first result gives a complete description of the three-body case.

\begin{theorem}[Three-body equilibria]
\label{thm:intro-three-body}
Consider the charged three-body problem with positive masses $m_1,m_2,m_3>0$
and real charges $e_1,e_2,e_3\in\mathbb R$.
Then under the non-trivial assumption, the following statements hold.

\begin{enumerate}
    \item Every non-collision equilibrium is collinear.

    \item A non-collision equilibrium exists if and only if there is a
    permutation \((i,j,k)\) of \((1,2,3)\) such that
    \[
        \delta_{ij}\delta_{jk}>0,
        \qquad
        \delta_{ij}\delta_{ik}<0,
    \]
    and
    \[
        \sqrt{|\delta_{ik}|}
        =
        \sqrt{|\delta_{ij}|}
        +
        \sqrt{|\delta_{jk}|}.
    \]
    In this case, the \(j\)-th particle lies between the \(i\)-th and the
    \(k\)-th particles, and the adjacent distances satisfy
    \[
        |a_i-a_j|:|a_j-a_k|
        =
        \sqrt{|\delta_{ij}|}:\sqrt{|\delta_{jk}|}.
        \]
    The equilibrium is unique up to translations, rotations and scalings.

    \item Every such equilibrium is linearly unstable. 
\end{enumerate}
\end{theorem}
\begin{proof}
The first assertion follows from Lemma~\ref{Lm:c3bp}. 
The existence criterion, the
position of the middle particle, the distance ratio, and the uniqueness up to
translations, rotations and scalings are exactly Theorem~\ref{thm:charged-three-body-equilibrium}. 
Finally, the
linear instability of all such equilibria follows from Proposition~\ref{prop:3bd.unstable}.
\end{proof}

The four-body problem has two rather different aspects. Non-collinear
four-body equilibria are constrained by strong geometric and sign restrictions,
whereas collinear four-body equilibria are flexible from the viewpoint of the
inverse problem. We summarize both aspects in the following two theorems.

\begin{theorem}[Four-body equilibria]
\label{thm:intro-four-body}
Consider the charged four-body problem with positive masses $m_1,m_2,m_3,m_4>0$
and real charges $e_1,e_2,e_3,e_4\in\mathbb R$.
Then under the non-trivial assumption, the following statements hold.

\begin{enumerate}
    \item Every four-body equilibrium is coplanar. Moreover, if the equilibrium
    is non-collinear, then no three of the four particles are collinear.

    \item Suppose that $a_1,a_2,a_3,a_4$
    form a convex quadrilateral and are labelled counterclockwise along its
    boundary. Then a non-trivial equilibrium can occur only if
    \[
        \delta_{12}>0,\qquad
        \delta_{23}>0,\qquad
        \delta_{34}>0,\qquad
        \delta_{14}>0,
    \]
    and
    \[
        \delta_{13}<0,\qquad
        \delta_{24}<0.
    \]
    Thus the four boundary edges correspond to attractive effective
    interactions, while the two diagonals correspond to repulsive effective
    interactions.

    \item Suppose that the configuration is concave and that \(a_4\) lies in
    the interior of the triangle formed by \(a_1,a_2,a_3\). Then the three
    coefficients on the outer triangle have the same sign, while the three
    coefficients connecting the interior point to the outer vertices have the
    opposite sign. More precisely, either
    \[
        \delta_{12}>0,\qquad
        \delta_{13}>0,\qquad
        \delta_{23}>0,
    \]
    and
    \[
        \delta_{14}<0,\qquad
        \delta_{24}<0,\qquad
        \delta_{34}<0,
    \]
    or the reverse sign pattern holds.

    \item A non-trivial four-body equilibrium cannot form a concyclic
    quadrilateral. In particular, a square equilibrium does not exist.

    \item On the other hand, every prescribed collinear configuration of four
    distinct points can be realized as a non-trivial equilibrium after choosing
    suitable positive masses and real charges. More precisely, let
    \[
        a_1<a_2<a_3<a_4
    \]
    be prescribed. For a fixed positive mass vector $m=(m_1,m_2,m_3,m_4)\in(\mathbb R^+)^4$,
    define
    \[
    \begin{aligned}
    \mathcal D(a,m)
    =
    &-\frac{m_1^2}
    {(a_1-a_2)^2(a_1-a_3)^2(a_1-a_4)^2}
    +\frac{m_2^2}
    {(a_2-a_1)^2(a_2-a_3)^2(a_2-a_4)^2}
    \\
    &-\frac{m_3^2}
    {(a_3-a_1)^2(a_3-a_2)^2(a_3-a_4)^2}
    +\frac{m_4^2}
    {(a_4-a_1)^2(a_4-a_2)^2(a_4-a_3)^2}.
    \end{aligned}
    \]
    If $\mathcal D(a,m)\neq0$,
    then there exist real charges $e_1,e_2,e_3,e_4$
    such that the prescribed collinear configuration is a non-trivial
    equilibrium of the charged four-body problem. In particular, for each fixed
    collinear configuration \(a\), all positive masses outside a proper
    algebraic subset satisfy this sufficient condition.
\end{enumerate}
\end{theorem}
\begin{proof}
The coplanarity assertion follows from Lemma~\ref{lem:coplanar.4bp}, 
and the statement that no
three particles are collinear in the non-collinear case follows from
Lemma~\ref{lem:no-three-collinear-four-body}. 
Assertions 2 and 3, i.e., the sign restrictions for convex and concave quadrilaterals are
proved in Proposition~\ref{prop:four-body-sign-patterns}. 
Assertion 4, concerning the obstruction to concyclic quadrilaterals, follows from Corollary~\ref{cor:degenerate-auxiliary-triangle}; 
in particular, a square equilibrium is impossible under the
non-trivial assumption. 
Finally, the last assertion, the collinear realization statement for four
prescribed points, is Proposition~\ref{prop:N4-rank-criterion}.
\end{proof}

\begin{theorem}[A necessary condition for non-collinear four-body equilibria]
\label{thm:intro-four-body}
Consider the charged four-body problem with positive masses $m_1,m_2,m_3,m_4>0$
and real charges $e_1,e_2,e_3,e_4\in\mathbb R$.
Suppose that the four particles form a non-collinear equilibrium,
labelled as
\[
    A=a_1,\qquad B=a_2,\qquad C=a_3,\qquad D=a_4.
\]
Define
\[
    a=|\delta_{14}\delta_{23}|^{1/3},
    \qquad
    b=|\delta_{12}\delta_{34}|^{1/3},
    \qquad
    c=|\delta_{13}\delta_{24}|^{1/3}.
\]
Then under the non-trivial assumption, \(a,b,c\) must be the side lengths of a possibly non-degenerate triangle. In
particular,
\[
    a< b+c,\qquad b< c+a,\qquad c< a+b.
\]

Moreover, the corresponding
geometric data give rise to three homogeneous polynomial equations in three
variables. Consequently, the associated resultant must vanish:
\[
    \operatorname{Res}\bigl(F_0,F_1,F_2\bigr)=0,
\]
where the polynomials $F_0,\;F_1,\;F_2$
are given explicitly in \eqref{F0:four-body-polynomials}-\eqref{F2:four-body-polynomials} below.
\end{theorem}
\begin{proof}
By Theorem~\ref{thm:charged-four-body-equilibrium}, 
the three quantities $a=|\delta_{14}\delta_{23}|^{1/3}, b=|\delta_{12}\delta_{34}|^{1/3}$
and $c=|\delta_{13}\delta_{24}|^{1/3}$
are the side lengths of a possibly degenerate auxiliary triangle. 
However, the
degenerate case cannot occur under the non-trivial assumption. Indeed, as shown
in the proof of Corollary~\ref{cor:degenerate-auxiliary-triangle}, degeneration of the auxiliary triangle would
imply that the four points \(A,B,C,D\) are concyclic, while Corollary~\ref{cor:degenerate-auxiliary-triangle}
excludes non-trivial concyclic four-body equilibria. Hence the auxiliary
triangle is non-degenerate, and therefore $a<b+c,\; b<c+a,\; c<a+b$.

The remaining assertion is precisely the non-degenerate part of
Theorem~\ref{thm:charged-four-body-equilibrium}. 
\end{proof}

We also obtain a general result for collinear configurations.

\begin{theorem}[Collinear equilibria and the inverse problem]
\label{thm:intro-collinear}
Let \(N\geq3\). For a prescribed collinear ordering
\[
    a_1<a_2<\cdots<a_N,
\]
set $x_i=a_{i+1}-a_i$ for $1\leq i\leq N-1$. Then

\begin{enumerate}
    \item \emph{A resultant necessary condition.}
    Fix positive masses $m=(m_1,\ldots,m_N)\in(\mathbb R^+)^N$
    and real charges $e=(e_1,\ldots,e_N)\in\mathbb R^N$.
    Under the nontrivial assumption, if a non-collision collinear equilibrium exists, then the adjacent distances
    \[
        (x_1,\ldots,x_{N-1})\in(\mathbb R^+)^{N-1}
    \]
    give a positive common zero of the homogeneous polynomial system
    \[
        G_1(x_1,\ldots,x_{N-1})
        =
        \cdots
        =
        G_{N-1}(x_1,\ldots,x_{N-1})
        =
        0,
    \]
    where the polynomials \(G_i\) are defined in
    \eqref{eq:collinear-Gi} below. Consequently, 
    the associated resultant must vanish:
    \[
        \operatorname{Res}(G_1,\ldots,G_{N-1})=0.
    \]

    \item \emph{The inverse problem.}
    Conversely, for every prescribed collinear configuration $a_1<a_2<\cdots<a_N$
    of \(N\) distinct points, one can choose positive masses $m_1,\ldots,m_N>0$
    and real charges $e_1,\ldots,e_N\in\mathbb R$
    such that the prescribed configuration is a non-trivial equilibrium of
    the charged \(N\)-body problem.
\end{enumerate}
\end{theorem}
\begin{proof}
The first part is Theorem~\ref{thm:necessary.condition.collinear.nbp}. 
After the denominators in the collinear
force-balance equations are cleared, the equilibrium condition gives the
homogeneous polynomial system $G_1=\cdots=G_{N-1}=0$.
Thus the existence of a positive common zero implies the vanishing of the
resultant. The inverse statement is Theorem~\ref{thm:realization-collinear}, 
which shows that every
prescribed collinear configuration of \(N\) distinct points can be realized as
a non-trivial equilibrium by choosing suitable positive masses and real
charges.
\end{proof}

Finally, we study regular polygon configurations.

\begin{theorem}[Regular polygon configurations]
\label{thm:intro-regular-polygon}
Let \(n\geq3\). Place \(n\) particles at the vertices of a regular \(n\)-gon.
Then the equilibrium equations can be written in the matrix form
\[
    (A-\Lambda A\Lambda)m=0,
    \qquad
    (B-\Lambda B\Lambda)m=0,
\]
where \(A\) and \(B\), given by \eqref{eq:1} below, are respectively a
symmetric circulant matrix and a skew-symmetric circulant matrix, and
\begin{eqnarray*}
    \Lambda&=&\operatorname{diag}(\lambda_1,\ldots,\lambda_n),
    \qquad
    \lambda_i=\frac{e_i}{m_i},
    \\
    m&=&(m_1,m_2,\ldots,m_n)^T.
\end{eqnarray*}

If the masses on the vertices are equal and the configuration is an
equilibrium, then
\[
    \lambda_1=\cdots=\lambda_n=\pm1.
\]
Equivalently,
\[
    |e_1|=\cdots=|e_n|=m_1=\cdots=m_n.
\]
Therefore, under the non-trivial assumption, equal masses cannot form a regular
\(n\)-gon equilibrium.

For \(n=4\), the conclusion is stronger: a square equilibrium is necessarily
trivial for arbitrary positive masses.

If one additional particle with mass \(m_0\) and charge \(e_0\) is placed at
the center of the regular \(n\)-gon, and if the \(n\) particles on the vertices
have the same mass
\[
    m_*=m_1=\cdots=m_n>0,
\]
then the centered regular \(n\)-gon is an equilibrium if and only if
\[
    \lambda_1=\cdots=\lambda_n=:\lambda
\]
and
\[
    m_*(1-\lambda^2)\sum_{i=1}^n a_i=e_0\lambda-m_0.
\]
\end{theorem}
\begin{proof}
The matrix form of the equilibrium equations for a regular \(n\)-gon is
Proposition~\ref{prop:equilibrium.eq.of.n-gon}. 
The equal-mass conclusion follows from Theorem~\ref{thm:nonexistence.of.n-gon}. 
The stronger square statement follows from Corollary \ref{cor:degenerate-auxiliary-triangle}. 
Finally, the characterization of the
centered regular \(n\)-gon with equal outer masses follows from Theorem~\ref{thm:centered-regular-ngon}.
\end{proof}

\subsection{Finiteness problem for equilibria}

The finiteness problem is one of the classical questions in the Newtonian
\(N\)-body problem. In the purely gravitational case, one usually studies this
question for central configurations rather than for static equilibria, because
non-collision static equilibria do not exist. 
The classical
finiteness problem asks whether, for fixed positive masses, there are only
finitely many central configurations modulo these natural symmetries.

This problem goes back at least to Wintner and was later included in Smale's
list of mathematical problems for the twenty-first century; see
\cite{Wintner1941} and \cite{Smale1998}. Important progress has been made in low
dimensions and for small numbers of bodies. For instance, Hampton and Moeckel
proved the finiteness of relative equilibria in the Newtonian four-body
problem, and Albouy and Kaloshin proved a generic finiteness theorem for
planar five-body central configurations; see
\cite{HamptonMoeckel2006} and \cite{AlbouyKaloshin2012} and references therein.

For the charged \(N\)-body problem, the situation is different. The analogue
of the finiteness statement for central configurations is no longer true in
general. 
In fact, in \cite{AlP2002}, Alfaro and
Pérez-Chavela constructed families of continua of central configurations in
charged problems, showing that the finiteness phenomenon familiar from the
Newtonian problem may break down once electrostatic interactions are included. 
We recall the following elementary example,
in the notation of the present paper, in order to illustrate this mechanism and
to motivate the corresponding finiteness question for static equilibria.
Indeed, consider four particles with
\[
    m_1=m_2=m_3=m_4=1,
    \qquad
    e_1=e_2=1,\quad e_3=e_4=-1.
\]
Then $\delta_{12}=\delta_{34}=0$ and $\delta_{13}=\delta_{14}=\delta_{23}=\delta_{24}=2$.
For any \(a,b>0\), set
\[
    q_1=(a,0),\qquad q_2=(-a,0),
    \qquad
    q_3=(0,b),\qquad q_4=(0,-b).
\]
Let \(r=\sqrt{a^2+b^2}\). A direct computation gives
\[
    \frac{\partial U}{\partial q_i}(q)
    =
    -\frac{4}{r^3}q_i,
    \qquad 1\leq i\leq4.
\]
Hence
\[
    \frac{4}{r^3}q_i+\frac{\partial U}{\partial q_i}(q)=0,
    \qquad 1\leq i\leq4,
\]
so \(q=(q_1,q_2,q_3,q_4)\) is a central configuration. After quotienting out
translations, rotations and scalings, the ratio \(a/b\) remains a free
parameter. Thus the charged four-body problem admits a continuous family of
central configurations for fixed masses and charges.

This example shows that the classical finiteness problem for central
configurations cannot be directly transferred to the charged \(N\)-body
problem. However, it does not settle the corresponding question for static
equilibria. Since a static equilibrium is defined by
\[
    \frac{\partial U}{\partial q}=0,
\]
it is natural to ask whether the set of such equilibria is finite, after
quotienting out translations, rotations and scalings.

\begin{question}
For fixed positive masses $m_1,\ldots,m_N>0$
and fixed real charges $e_1,\ldots,e_N\in\mathbb R$,
assume that the charged \(N\)-body problem admits at least one non-collision
equilibrium. Is the number of non-collision equilibria finite, up to
translations, rotations and scalings?
\end{question}

The results of this paper can be viewed as first steps toward this problem.
For three bodies, the equilibria are completely classified. For four bodies,
we obtain strong geometric restrictions in the non-collinear case and explicit
algebraic criteria in several collinear and symmetric reductions.
\medskip

The paper is organized as follows. Section~2 recalls the necessary facts about
multivariate resultants. Section~3 gives the complete classification and
linear instability of three-body equilibria. Section~4 studies non-collinear
four-body equilibria and derives their geometric restrictions. In addition to
the general necessary conditions, we also discuss two symmetric reductions: an
equilateral-triangle-with-center configuration and a kite-shaped configuration.
Section~5 treats collinear equilibria and the corresponding inverse problem,
including a special collinear four-body reduction with
\(\delta_{14}=\delta_{23}=0\) and a more explicit four-body realization
criterion. Section~6 studies regular polygon configurations and centered
regular polygon configurations. The appendix contains the proofs of some
auxiliary geometric and spectral results.

\section{Preliminaries: multivariate resultants}
In this subsection we recall several standard facts about multivariate
resultants. 
We refer to \cite{CLO, HaM2015} and the references therein for more
details.

Let \(K\) be an algebraically closed field and let
\(R=K[x_0,\ldots,x_n].\)
Suppose that
\(F_0,\ldots,F_n\in R\)
are homogeneous polynomials of degrees \(d_0,\ldots,d_n\),
respectively. We write
\[
    F_i=\sum_{|\alpha|=d_i}c_{i,\alpha}x^\alpha,
    \qquad i=0,\ldots,n.
\]
Here the coefficients \(c_{i,\alpha}\) may be regarded either as fixed elements
of \(K\), or as independent variables when one considers the universal
polynomial system.

\begin{theorem}[\cite{CLO}, Chapter 3, Theorem 2.3]\label{thm:multivariate-resultant}
For fixed degrees \(d_0,\ldots,d_n\), there exists a unique polynomial
\(\operatorname{Res}_{d_0,\ldots,d_n}\in \mathbb Z[u_{i,\alpha}]\)
with the following properties:
\begin{enumerate}
    \item If \(F_0,\ldots,F_n\in K[x_0,\ldots,x_n]\) are homogeneous of degrees
    \(d_0,\ldots,d_n\), respectively, then the system
    \[
        F_0=\cdots=F_n=0
    \]
    has a non-trivial solution over \(K\) if and only if
    \(\operatorname{Res}(F_0,\ldots,F_n)=0\).

    \item The resultant is normalized by
    \(\operatorname{Res}(x_0^{d_0},\ldots,x_n^{d_n})=1\).

    \item The polynomial \(\operatorname{Res}_{d_0,\ldots,d_n}\) is irreducible
    in the coefficient variables.
\end{enumerate}
\end{theorem}

Here a non-trivial solution means a point
\((x_0,\ldots,x_n)\neq 0\)
satisfying all the equations. Equivalently, the vanishing of the resultant
detects whether the hypersurfaces defined by \(F_0,\ldots,F_n\) have a common
point in the projective space \(\mathbb P^n(K)\).

We shall also use the following two algebraic properties.

\begin{theorem}[\cite{HaM2015}, Theorem 2]
\label{thm:resultant-properties}
Let \(F_0,\ldots,F_n\in K[x_0,\ldots,x_n]\) be homogeneous polynomials of
degrees \(d_0,\ldots,d_n\), respectively. Then the following hold.

\begin{enumerate}
    \item If \(i<j\), then
    \[
    \begin{aligned}
    \operatorname{Res}
    (F_0,\ldots,F_i,\ldots,F_j,\ldots,F_n)  
    =(-1)^{d_0d_1\cdots d_n}
    \operatorname{Res}
    (F_0,\ldots,F_j,\ldots,F_i,\ldots,F_n).
    \end{aligned}
    \]

    \item If \(F_j=F'_jF''_j\),
    where \(F'_j\) and \(F''_j\) are homogeneous polynomials, then
    \[
    \begin{aligned}
    \operatorname{Res}(F_0,\ldots,F_j,\ldots,F_n)
    =
    \operatorname{Res}(F_0,\ldots,F'_j,\ldots,F_n)\cdot
    \operatorname{Res}(F_0,\ldots,F''_j,\ldots,F_n).
    \end{aligned}
    \]
\end{enumerate}
\end{theorem}

The multiplicativity property explains why, in our later computations, the
resultant may split into several algebraic factors. Each factor gives a
possible algebraic branch for the existence of a non-zero projective solution.
However, in the charged \(N\)-body problem, the projective variables often
represent mutual distances or ratios of distances. Thus, after obtaining the
condition
\[
    \operatorname{Res}(F_0,\ldots,F_n)=0,
\]
one must still determine which algebraic branches admit positive real solutions.

We now recall a concrete formula which will be useful for the case of three
ternary quadratic forms. Let
\[
\begin{aligned}
F_0={}&c_{01}X^2+c_{02}Y^2+c_{03}Z^2
      +c_{04}XY+c_{05}XZ+c_{06}YZ,\\
F_1={}&c_{11}X^2+c_{12}Y^2+c_{13}Z^2
      +c_{14}XY+c_{15}XZ+c_{16}YZ,\\
F_2={}&c_{21}X^2+c_{22}Y^2+c_{23}Z^2
      +c_{24}XY+c_{25}XZ+c_{26}YZ.
\end{aligned}
\]
These are homogeneous quadratic polynomials in three variables. By Theorem
\ref{thm:multivariate-resultant},
\[
    \operatorname{Res}(F_0,F_1,F_2)=0
\]
if and only if the system
\[
    F_0=F_1=F_2=0
\]
has a non-trivial common zero in \(\mathbb P^2(K)\). In principle, this
resultant is a polynomial in the eighteen coefficients \(c_{ij}\). The following
formula gives an efficient determinant expression for it.

Let
\[
    J=\det
    \begin{pmatrix}
    \dfrac{\partial F_0}{\partial X} &
    \dfrac{\partial F_0}{\partial Y} &
    \dfrac{\partial F_0}{\partial Z} \\[1.2ex]
    \dfrac{\partial F_1}{\partial X} &
    \dfrac{\partial F_1}{\partial Y} &
    \dfrac{\partial F_1}{\partial Z} \\[1.2ex]
    \dfrac{\partial F_2}{\partial X} &
    \dfrac{\partial F_2}{\partial Y} &
    \dfrac{\partial F_2}{\partial Z}
    \end{pmatrix}
\]
be the Jacobian determinant of \(F_0,F_1,F_2\). Since \(J\) is a homogeneous
cubic polynomial, its first partial derivatives are homogeneous quadratic
polynomials. Write
\[
\begin{aligned}
\frac{\partial J}{\partial X}
={}&b_{01}X^2+b_{02}Y^2+b_{03}Z^2
   +b_{04}XY+b_{05}XZ+b_{06}YZ,\\
\frac{\partial J}{\partial Y}
={}&b_{11}X^2+b_{12}Y^2+b_{13}Z^2
   +b_{14}XY+b_{15}XZ+b_{16}YZ,\\
\frac{\partial J}{\partial Z}
={}&b_{21}X^2+b_{22}Y^2+b_{23}Z^2
   +b_{24}XY+b_{25}XZ+b_{26}YZ.
\end{aligned}
\]

\begin{proposition}[\cite{HaM2015}, Proposition 1]\label{prop:three-ternary-quadrics}
With the notation above, the resultant of \(F_0,F_1,F_2\) is given by
\[
\operatorname{Res}(F_0,F_1,F_2)
=
\frac{-1}{512}
\det
\begin{pmatrix}
c_{01} & c_{02} & c_{03} & c_{04} & c_{05} & c_{06} \\
c_{11} & c_{12} & c_{13} & c_{14} & c_{15} & c_{16} \\
c_{21} & c_{22} & c_{23} & c_{24} & c_{25} & c_{26} \\
b_{01} & b_{02} & b_{03} & b_{04} & b_{05} & b_{06} \\
b_{11} & b_{12} & b_{13} & b_{14} & b_{15} & b_{16} \\
b_{21} & b_{22} & b_{23} & b_{24} & b_{25} & b_{26}
\end{pmatrix}.
\]
\end{proposition}

If $F_0,F_1,F_2$ have simpler forms, we have
\begin{proposition}\label{prop:special-ternary-quadrics}
Suppose
\begin{eqnarray*}
F_0&=&X^2 + b_1 XY + c_1 Y^2, \\
F_1&=&Y^2 + b_2 YZ + c_2 Z^2, \\
F_2&=&Z^2 + b_3 ZX + c_3 X^2.
\end{eqnarray*}
Then the resultant of \(F_0,F_1,F_2\) is given by
\begin{eqnarray*}
\operatorname{Res}(F_0,F_1,F_2)
=\prod_{\varepsilon_1,\varepsilon_2,\varepsilon_3=\pm1}
\left[
1-\frac{(-b_1+\epsilon_1\sqrt{b_1^2-4c_1})(-b_2+\epsilon_2\sqrt{b_2^2-4c_2})(-b_3+\epsilon_3\sqrt{b_3^2-4c_3})}{8}
\right].
\end{eqnarray*}
\end{proposition}

In the sequel, some equilibrium equations will be reduced to systems of three
homogeneous quadratic equations in three variables. Proposition
\ref{prop:three-ternary-quadrics}, and in particular the special form in
Proposition \ref{prop:special-ternary-quadrics}, provides a practical way to
compute their resultants. After the resultant is factored, Theorem
\ref{thm:multivariate-resultant} gives an algebraic criterion for the existence
of a non-zero projective common root. Finally, because the variables in our
problem correspond to distances or ratios of distances, the resulting algebraic
conditions must be combined with the positivity requirement on the relevant
distance variables.

\section{Three-body problem}
In this section we give a complete classification of the equilibria 
and its linear stability in the charged three-body problem.

\subsection{The existence condition}

Recall that \(\delta_{ij}\) is the effective interaction coefficient between the
\(i\)-th and the \(j\)-th particles. The case \(\delta_{ij}>0\) corresponds to an
attractive interaction, while \(\delta_{ij}<0\) corresponds to a repulsive
interaction.

We first observe that no non-collinear equilibrium can occur.

\begin{lemma}\label{Lm:c3bp}
Let \(a=(a_1,a_2,a_3)\) be a non-collision equilibrium of the charged
three-body problem. Under the non-trivial assumption that at least two of the
coefficients \(\delta_{ij}\) are nonzero, the three points \(a_1,a_2,a_3\) must
be collinear.
\end{lemma}

\begin{proof}
Suppose, by contradiction, that \(a_1,a_2,a_3\) are not collinear. Then, for
each fixed \(i\), the two vectors \(a_j-a_i\) and \(a_k-a_i\), where
\(\{i,j,k\}=\{1,2,3\}\), are linearly independent. The equilibrium equation for
the \(i\)-th particle is
\[
    \sum_{j\neq i}
    \frac{\delta_{ij}}{|a_j-a_i|^3}(a_j-a_i)=0 .
\]
Since the two vectors in this equation are linearly independent, both
coefficients must vanish. Taking \(i=1\), we obtain
\[
    \delta_{12}=\delta_{13}=0.
\]
Taking \(i=2\), we further obtain \(\delta_{23}=0\). Hence all three
interaction coefficients vanish, which contradicts the non-trivial assumption.
Therefore every non-collision equilibrium must be collinear.
\end{proof}

We now determine the exact condition for the existence of a collinear
equilibrium.

\begin{theorem}\label{thm:charged-three-body-equilibrium}
Consider the charged three-body problem with masses
\(m_1,m_2,m_3>0\) and charges \(e_1,e_2,e_3\). Suppose that at least two of
\(\delta_{12},\delta_{23},\delta_{13}\) are nonzero. Then a non-collision
equilibrium exists if and only if there is a permutation \((i,j,k)\) of
\((1,2,3)\) such that
\[
    \delta_{ij}\delta_{jk}>0,\qquad
    \delta_{ij}\delta_{ik}<0,
\]
and
\[
    \sqrt{|\delta_{ik}|}
    =
    \sqrt{|\delta_{ij}|}
    +
    \sqrt{|\delta_{jk}|}.
    \tag{3.1}
\]
In this case, the equilibrium is collinear. More precisely, the \(j\)-th
particle lies between the \(i\)-th and the \(k\)-th particles, and the adjacent
distances satisfy
\[
    |a_i-a_j|:|a_j-a_k|
    =
    \sqrt{|\delta_{ij}|}:\sqrt{|\delta_{jk}|}.
    \tag{3.2}
\]
Conversely, if the above conditions hold, then placing the three particles on a
line with the distance ratio \((3.2)\) gives a non-collision equilibrium, unique
up to translations, rotations and scalings.
\end{theorem}

\begin{proof}
By the previous lemma, every equilibrium must be collinear. After relabelling
the particles, we may assume that the order on the line is
$a_1,a_2,a_3$.
Up to a translation, a rotation and a scaling, we may write
\[
    a_1=0,\qquad
    a_2=x\ell,\qquad
    a_3=(x+1)\ell,
\]
where \(x>0\) and \(\ell>0\). The equilibrium equations become
\[
    \frac{\delta_{12}}{x^2\ell^2}
    +
    \frac{\delta_{13}}{(x+1)^2\ell^2}
    =0,
    \tag{3.3}
\]
\[
    -\frac{\delta_{12}}{x^2\ell^2}
    +
    \frac{\delta_{23}}{\ell^2}
    =0,
    \tag{3.4}
\]
and
\[
    -\frac{\delta_{13}}{(x+1)^2\ell^2}
    -
    \frac{\delta_{23}}{\ell^2}
    =0.
    \tag{3.5}
\]
Only two of these three equations are independent, since the sum of the total
forces is identically zero.

From \((3.4)\), we obtain
\[
    x^2=\frac{\delta_{12}}{\delta_{23}}.
\]
Since \(x>0\), this implies $\delta_{12}\delta_{23}>0$.
Moreover, from \((3.5)\), we obtain
\[
    (x+1)^2=-\frac{\delta_{13}}{\delta_{23}},
\]
and hence $\delta_{13}\delta_{23}<0$.
Therefore \(\delta_{12}\) and \(\delta_{23}\) have the same sign, while
\(\delta_{13}\) has the opposite sign.

Taking positive square roots gives
\[
    x=\sqrt{\frac{|\delta_{12}|}{|\delta_{23}|}},
    \qquad
    x+1=\sqrt{\frac{|\delta_{13}|}{|\delta_{23}|}}.
\]
Thus
\[
    \sqrt{|\delta_{12}|}+\sqrt{|\delta_{23}|}
    =
    \sqrt{|\delta_{13}|}.
    \tag{3.6}
\]
This is exactly condition \((3.1)\) in the special ordering
\((i,j,k)=(1,2,3)\).

Conversely, suppose that, for the ordering \(a_1,a_2,a_3\), one has
\[
    \delta_{12}\delta_{23}>0,\qquad
    \delta_{12}\delta_{13}<0,
\]
and
\[
    \sqrt{|\delta_{13}|}
    =
    \sqrt{|\delta_{12}|}
    +
    \sqrt{|\delta_{23}|}.
\]
Set
\[
    x=\sqrt{\frac{|\delta_{12}|}{|\delta_{23}|}},
\]
and choose
\[
    a_1=0,\qquad
    a_2=x\ell,\qquad
    a_3=(x+1)\ell,
    \qquad \ell>0.
\]
Then
\[
    x+1=\sqrt{\frac{|\delta_{13}|}{|\delta_{23}|}},
\]
and a direct substitution into \((3.3)\)--\((3.5)\) shows that the total force
on each particle is zero. Hence this collinear configuration is an equilibrium.

For a general ordering of the three particles, the same argument applies after
replacing \((1,2,3)\) by the corresponding permutation \((i,j,k)\). This proves
the theorem.
\end{proof}

\begin{remark}
In the special ordering \(a_1,a_2,a_3\), Theorem
\ref{thm:charged-three-body-equilibrium} says that an equilibrium exists exactly
in one of the following two cases:
\[
    \delta_{12}>0,\qquad
    \delta_{23}>0,\qquad
    \delta_{13}<0,
\]
with
\begin{equation}\label{condition.of.e.1}
    \sqrt{\delta_{12}}+\sqrt{\delta_{23}}
    =
    \sqrt{-\delta_{13}},
\end{equation}
or
\[
    \delta_{12}<0,\qquad
    \delta_{23}<0,\qquad
    \delta_{13}>0,
\]
with
\begin{equation}\label{condition.of.e.2}
    \sqrt{-\delta_{12}}+\sqrt{-\delta_{23}}
    =
    \sqrt{\delta_{13}}.
\end{equation}
Equivalently, the two adjacent interactions must have the same sign, the
interaction between the two endpoints must have the opposite sign, and the
absolute values must satisfy the square-root relation above.
\end{remark}

The preceding theorem gives an explicit algebraic criterion for the existence
of a three-body equilibrium. 


Next, we will consider the stability of the equilibrium, if it exists.

\subsection{Linear stability of the three-body equilibria}

We now discuss the linear stability of the collinear equilibria obtained above.
The equations of motion are
\[
    m_i\ddot q_i=\frac{\partial U}{\partial q_i},
    \qquad i=1,2,3,
\]
and an equilibrium \(a=(a_1,a_2,a_3)\) satisfies
\[
    \frac{\partial U}{\partial q}(a)=0.
\]
Thus the linearized equation at \(a\) is
\[
    M\ddot \xi=D^2U(a)\xi,
\]
where $M=\operatorname{diag}(m_1I_2,m_2I_2,m_3I_2)$.

Assume, after a translation and a rotation, that the equilibrium is ordered as
\[
    a_1=0,\qquad a_2=r_1,\qquad a_3=r_1+r_2,
    \qquad r_1,r_2>0.
\]
Set
\[
    r_3=r_1+r_2.
\]
The equilibrium condition is equivalent to
\[
    \frac{\delta_{12}}{r_1^2}
    =
    \frac{\delta_{23}}{r_2^2}
    =
    -\frac{\delta_{13}}{r_3^2}
    =:\kappa ,
    \qquad \kappa\neq 0.
\]
Hence
\[
    \delta_{12}=\kappa r_1^2,\qquad
    \delta_{23}=\kappa r_2^2,\qquad
    \delta_{13}=-\kappa r_3^2.
\]

Write the perturbation as
\[
    \xi_i=(u_i,v_i),
\]
where \(u_i\) is the perturbation in the collinear direction and \(v_i\) is the
transverse perturbation. Since the configuration is collinear, the linearized
system splits into the longitudinal and transverse parts:
\[
    M_0\ddot u=2Gu,\qquad
    M_0\ddot v=-Gv,
\]
where $M_0=\operatorname{diag}(m_1,m_2,m_3)$,
and \(G\) is the weighted Laplacian matrix
\[
    G=
    \begin{pmatrix}
    c_{12}+c_{13} & -c_{12} & -c_{13}\\
    -c_{12} & c_{12}+c_{23} & -c_{23}\\
    -c_{13} & -c_{23} & c_{13}+c_{23}
    \end{pmatrix},
\]
with $c_{ij}=\frac{\delta_{ij}}{|a_i-a_j|^3}$.
Using the equilibrium relations, we have
\[
    c_{12}=\frac{\kappa}{r_1},\qquad
    c_{23}=\frac{\kappa}{r_2},\qquad
    c_{13}=-\frac{\kappa}{r_3}.
\]
Therefore, for any \(z=(z_1,z_2,z_3)^T\in\mathbb R^3\),
\[
\begin{aligned}
    z^TGz
    &=
    \kappa\left[
    \frac{(z_1-z_2)^2}{r_1}
    +
    \frac{(z_2-z_3)^2}{r_2}
    -
    \frac{(z_1-z_3)^2}{r_3}
    \right]  \\
    &=
    \kappa
    \frac{
    \bigl(r_2(z_1-z_2)-r_1(z_2-z_3)\bigr)^2
    }{r_1r_2r_3}.
\end{aligned}
\]
Thus \(G\) has rank one. More precisely,
\[
    G
    =
    \frac{\kappa}{r_1r_2r_3}
    \begin{pmatrix}
    r_2\\
    -r_3\\
    r_1
    \end{pmatrix}
    \begin{pmatrix}
    r_2 & -r_3 & r_1
    \end{pmatrix}.
\]
Consequently, the generalized eigenvalue problem
\[
    Gz=\mu M_0z
\]
has two zero eigenvalues and one nonzero eigenvalue
\[
    \mu=\kappa \rho,
\]
where
\[
    \rho=
    \frac{1}{r_1r_2r_3}
    \left(
    \frac{r_2^2}{m_1}
    +
    \frac{r_3^2}{m_2}
    +
    \frac{r_1^2}{m_3}
    \right)>0.
\]

It follows that the nonzero characteristic exponents of the longitudinal part
satisfy
\[
    \sigma^2=2\mu=2\kappa\rho,
\]
whereas those of the transverse part satisfy
\[
    \sigma^2=-\mu=-\kappa\rho.
\]
Building on the above, we obtain the following result, which shows that even when such equilibria exist in the charged three-body problem, they are invariably linearly unstable.

\begin{proposition}\label{prop:3bd.unstable}
Every non-collision equilibrium of the charged three-body problem is linearly
unstable.
More precisely:
\begin{enumerate}
    \item If
    \[
        \delta_{12}>0,\qquad \delta_{23}>0,\qquad \delta_{13}<0,
    \]
    then \(\kappa>0\). The longitudinal part has real characteristic exponents
    \[
        \sigma=\pm\sqrt{2\kappa\rho},
    \]
    and hence the equilibrium is linearly unstable in the collinear direction.

    \item If
    \[
        \delta_{12}<0,\qquad \delta_{23}<0,\qquad \delta_{13}>0,
    \]
    then \(\kappa<0\). The transverse part has real characteristic exponents
    \[
        \sigma=\pm\sqrt{|\kappa|\rho},
    \]
    and hence the equilibrium is linearly unstable in the transverse direction.
\end{enumerate}
\end{proposition}

\section{Four-body problem}
\label{sec:4}
Having reserved the analysis of collinear equilibria for a general \(N\) to a
later section, we now focus on non-collinear equilibria of the charged
four-body problem. We first record several elementary qualitative restrictions on
such equilibria.

\subsection{Qualitative restrictions on four-body equilibria}

By a similar argument in the beginning of the proof of Lemma \ref{Lm:c3bp},
we have
\begin{lemma}\label{lem:coplanar.4bp}
    An equilibrium of the charged four-body problem must be co-planar.
\end{lemma}

\begin{lemma}\label{lem:no-three-collinear-four-body}
Let \(a=(a_1,a_2,a_3,a_4)\) be a non-trivial equilibrium of the charged
four-body problem. If the four particles are not collinear, then no three of
them are collinear.
\end{lemma}

\begin{proof}
Suppose, by contradiction, that three particles are collinear while the fourth
one is not. Without loss of generality, assume that
\(a_1,a_2,a_3\) are collinear and \(a_4\) is not on this line.

For each \(i=1,2,3\), the forces exerted on \(a_i\) by the other two collinear
particles have no component perpendicular to the line through
\(a_1,a_2,a_3\). Hence the perpendicular component of the force exerted by
\(a_4\) must vanish. Since \(a_4\) is not on the line, this implies
\[
    \delta_{14}=\delta_{24}=\delta_{34}=0.
\]
Equivalently,
\[
    e_1e_4=m_1m_4,\qquad
    e_2e_4=m_2m_4,\qquad
    e_3e_4=m_3m_4.
\]
Since \(m_4>0\), we have \(e_4\neq 0\), and hence
\[
    \frac{e_1}{m_1}
    =
    \frac{e_2}{m_2}
    =
    \frac{e_3}{m_3}
    =
    \frac{m_4}{e_4}.
\]
It follows that
\[
    \delta_{ij}
    =
    m_i m_j-e_i e_j
    =
    m_i m_j
    \left(
        1-\frac{m_4^2}{e_4^2}
    \right),
    \qquad 1\leq i<j\leq 3.
\]
Thus \(\delta_{12},\delta_{13},\delta_{23}\) have the same sign, unless they
all vanish.

Now consider the leftmost particle among \(a_1,a_2,a_3\). The two forces acting
on it from the other two collinear particles point in the same direction if
\(\delta_{12},\delta_{13},\delta_{23}\) have the same nonzero sign. Therefore
they cannot balance. Hence we must have
\[
    \delta_{12}=\delta_{13}=\delta_{23}=0.
\]
Together with
\[
    \delta_{14}=\delta_{24}=\delta_{34}=0,
\]
this gives \(\delta_{ij}=0\) for all \(1\leq i<j\leq 4\), contradicting the
non-trivial assumption. Therefore no three particles can be collinear.
\end{proof}

\begin{proposition}\label{prop:four-body-sign-patterns}
Let \(a=(a_1,a_2,a_3,a_4)\) be a non-trivial non-collinear equilibrium of the
charged four-body problem.

\begin{enumerate}
    \item If \(a_1,a_2,a_3,a_4\) form a convex quadrilateral and are labelled
    counterclockwise along its boundary, then the necessary sign pattern is
    \[
        \delta_{12}>0,\;
        \delta_{23}>0,\;
        \delta_{34}>0,\;
        \delta_{14}>0\;\;{\rm and}\;\;\delta_{13}<0,\;
        \delta_{24}<0.
    \]
    In other words, the four sides must correspond to attractive effective
    interactions, while the two diagonals must correspond to repulsive effective
    interactions.

    \item If the configuration is concave, and if \(a_4\) lies in the interior of
    the triangle formed by \(a_1,a_2,a_3\), then the three coefficients on the
    outer triangle have the same sign, while the three coefficients connecting
    the interior point to the outer vertices have the opposite sign. More
    precisely, either
    \[
        \delta_{12}>0,\;
        \delta_{13}>0,\;
        \delta_{23}>0
        \;\;{\rm and}\;\;\delta_{14}<0,\;
        \delta_{24}<0,\;
        \delta_{34}<0,
    \]
    or
    \[
        \delta_{12}<0,\;
        \delta_{13}<0,\;
        \delta_{23}<0
        \;\;{\rm and}\;\;
        \delta_{14}>0,\;
        \delta_{24}>0,\;
        \delta_{34}>0.
    \]
\end{enumerate}
\end{proposition}
\begin{proof}
For convenience, write
\[
    c_{ij}=\frac{\delta_{ij}}{|a_i-a_j|^3}.
\]
Clearly \(c_{ij}\) and \(\delta_{ij}\) have the same sign. The equilibrium
equation for the \(i\)-th particle is
\[
    \sum_{j\neq i} c_{ij}(a_j-a_i)=0.
\]

First suppose that \(a_1,a_2,a_3,a_4\) form a convex quadrilateral, labelled
counterclockwise. At the vertex \(a_1\), the vector \(a_3-a_1\) lies in the
interior of the cone generated by \(a_2-a_1\) and \(a_4-a_1\). Thus there exist
\(\alpha,\beta>0\) such that
\[
    a_3-a_1
    =
    \alpha(a_2-a_1)+\beta(a_4-a_1).
\]
The force balance at \(a_1\) is
\[
    c_{12}(a_2-a_1)
    +
    c_{13}(a_3-a_1)
    +
    c_{14}(a_4-a_1)
    =
    0.
\]
Substituting the previous relation gives
\[
    (c_{12}+\alpha c_{13})(a_2-a_1)
    +
    (c_{14}+\beta c_{13})(a_4-a_1)
    =
    0.
\]
Since \(a_2-a_1\) and \(a_4-a_1\) are linearly independent, we get
\[
    c_{12}=-\alpha c_{13},\qquad
    c_{14}=-\beta c_{13}.
\]
Hence
\[
    \operatorname{sgn}(\delta_{12})
    =
    \operatorname{sgn}(\delta_{14})
    =
    -\operatorname{sgn}(\delta_{13}).
\]
Applying the same argument at the other three vertices gives
\[
    \operatorname{sgn}(\delta_{12})
    =
    \operatorname{sgn}(\delta_{23})
    =
    \operatorname{sgn}(\delta_{34})
    =
    \operatorname{sgn}(\delta_{14}),
\]
and
\[
    \operatorname{sgn}(\delta_{13})
    =
    \operatorname{sgn}(\delta_{24})
    =
    -\operatorname{sgn}(\delta_{12}).
\]
Thus, purely from the geometry of force balance, one obtains either
\[
    (\delta_{12},\delta_{23},\delta_{34},\delta_{14},\delta_{13},\delta_{24})
    =
    (+,+,+,+,-,-),
\]
or
\[
    (\delta_{12},\delta_{23},\delta_{34},\delta_{14},\delta_{13},\delta_{24})
    =
    (-,-,-,-,+,+).
\]

We now show that the second possibility is incompatible with
\(\delta_{ij}=m_i m_j-e_i e_j\). If
\[
    \delta_{12},\delta_{23},\delta_{34},\delta_{14}<0,
    \qquad
    \delta_{13},\delta_{24}>0,
\]
then
\[
    e_1e_2>m_1m_2,\qquad
    e_3e_4>m_3m_4,
\]
and hence
\[
    e_1e_2e_3e_4>m_1m_2m_3m_4.
\]
On the other hand, since all four adjacent products are positive, the charges
\(e_1,e_2,e_3,e_4\) must have the same sign. Therefore \(e_1e_3>0\) and
\(e_2e_4>0\). But \(\delta_{13}>0\) and \(\delta_{24}>0\) imply
\[
    0<e_1e_3<m_1m_3,
    \qquad
    0<e_2e_4<m_2m_4.
\]
Multiplying these two inequalities gives
\[
    e_1e_2e_3e_4<m_1m_2m_3m_4,
\]
a contradiction. Hence the only possible convex sign pattern is
\[
    \delta_{12},\delta_{23},\delta_{34},\delta_{14}>0,
    \qquad
    \delta_{13},\delta_{24}<0.
\]

Now suppose that the configuration is concave and that \(a_4\) lies inside the
triangle formed by \(a_1,a_2,a_3\). At the vertex \(a_1\), the vector
\(a_4-a_1\) lies in the interior of the cone generated by \(a_2-a_1\) and
\(a_3-a_1\). The same cone argument as above gives
\[
    \operatorname{sgn}(\delta_{12})
    =
    \operatorname{sgn}(\delta_{13})
    =
    -\operatorname{sgn}(\delta_{14}).
\]
Similarly, the force balance at \(a_2\) gives
\[
    \operatorname{sgn}(\delta_{12})
    =
    \operatorname{sgn}(\delta_{23})
    =
    -\operatorname{sgn}(\delta_{24}),
\]
and the force balance at \(a_3\) gives
\[
    \operatorname{sgn}(\delta_{13})
    =
    \operatorname{sgn}(\delta_{23})
    =
    -\operatorname{sgn}(\delta_{34}).
\]
Combining these three relations yields
\[
    \operatorname{sgn}(\delta_{12})
    =
    \operatorname{sgn}(\delta_{13})
    =
    \operatorname{sgn}(\delta_{23}),
\]
and
\[
    \operatorname{sgn}(\delta_{14})
    =
    \operatorname{sgn}(\delta_{24})
    =
    \operatorname{sgn}(\delta_{34})
    =
    -\operatorname{sgn}(\delta_{12}).
\]
This proves the asserted sign pattern in the concave case.
\end{proof}

\subsection{A geometric formulation of four-body equilibria}
\label{subsec:4.2}

Before stating the general proportional relation, we first describe the
geometry of a four-body static equilibrium.
Suppose the four particles located at $a_1,a_2,a_3,a_4$ are labeled 
$A,B,C,D$, respectively, and form a static equilibrium (see Figure \ref{fig:four-body-static-equilibrium}).
Recall that, for each pair of
particles, the effective interaction is determined by
\[
    \delta_{ij}=m_i m_j-e_i e_j.
\]
In the geometric formulation below, only the magnitudes of the six pairwise
forces enter the proportional relations. Hence we write
\[
    \Delta_{ij}=|\delta_{ij}|,\qquad 1\leq i<j\leq 4.
\]
Denote by \(AB,BC,CD,DA,AC,BD\) the six mutual distances. Moreover, let
\[
    R_A=R_{BCD},\qquad
    R_B=R_{ACD},\qquad
    R_C=R_{ABD},\qquad
    R_D=R_{ABC},
\]
where \(R_{BCD}\), for example, is the circumradius of the triangle formed by
\(B,C,D\). Thus \(R_A\) is the circumradius of the triangle opposite to the
vertex \(A\), and similarly for the other vertices.

\begin{figure}[htbp]
\centering
\begin{tikzpicture}[scale=1.1]
    \coordinate (A) at (0,0);
    \coordinate (B) at (3.5,0.0);
    \coordinate (C) at (2.4,2.8);
    \coordinate (D) at (0.2,1.7);

    \draw[thick] (A)--(B)--(C)--(D)--cycle;
    \draw[dashed] (A)--(C);
    \draw[dashed] (B)--(D);

    \fill (A) circle (2pt) node[below left] {$A=a_1$};
    \fill (B) circle (2pt) node[below right] {$B=a_2$};
    \fill (C) circle (2pt) node[above right] {$C=a_3$};
    \fill (D) circle (2pt) node[above left] {$D=a_4$};

    \node at (1.7,-0.2) {$AB$};
    \node at (3.25,1.45) {$BC$};
    \node at (1.0,2.35) {$CD$};
    \node at (-0.25,0.85) {$DA$};
    \node at (1.55,1.55) {$AC$};
    \node at (1.45,0.9) {$BD$};
\end{tikzpicture}
\caption{A schematic four-body static equilibrium. The four boundary edges and
the two diagonals represent the six pairwise interactions.}
\label{fig:four-body-static-equilibrium}
\end{figure}
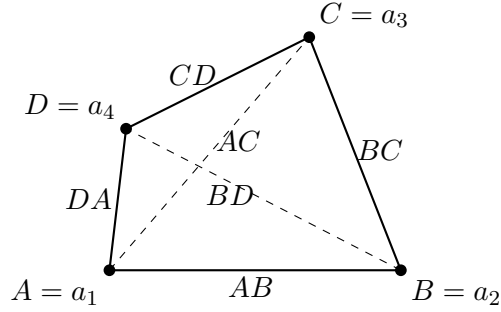

The following theorem gives a necessary geometric relation for a non-collinear
four-body static equilibrium. It concerns only the magnitudes of the effective
interactions; the corresponding sign conditions will be discussed separately.

\begin{lemma}[Geometric characterization of four-body equilibrium]\label{lem:geometric-characterization}
Let \(A,B,C,D\) be a non-collinear static equilibrium of the charged four-body
problem. With the notation introduced above, the following proportional relation
holds:
\[
\begin{aligned}
&\frac{AB}{R_A R_B}:
\frac{BC}{R_B R_C}:
\frac{CD}{R_C R_D}:
\frac{DA}{R_D R_A}:
\frac{AC}{R_A R_C}:
\frac{BD}{R_B R_D}  \\
&=
\Delta_{12}^{2/3}\Delta_{34}^{-1/3}:
\Delta_{23}^{2/3}\Delta_{14}^{-1/3}:
\Delta_{34}^{2/3}\Delta_{12}^{-1/3}:
\Delta_{14}^{2/3}\Delta_{23}^{-1/3}:
\Delta_{13}^{2/3}\Delta_{24}^{-1/3}:
\Delta_{24}^{2/3}\Delta_{13}^{-1/3}.
\end{aligned}
\]
\end{lemma}
\begin{proof}
    The proof of this lemma will be deferred to the Appendix.
\end{proof}

Based on the above lemma, we have the following theorem.

\begin{theorem}\label{thm:charged-four-body-equilibrium}
Consider the charged four-body problem with masses
\(m_1,m_2,m_3,m_4>0\) and charges \(e_1,e_2,e_3,e_4\). Suppose that at least two of the quantities
\(\delta_{ij}=m_im_j-e_ie_j,1\le i<j\le4\) are nonzero,
and \(A=a_1\), \(B=a_2\), \(C=a_3\), \(D=a_4\) form a
non-collinear equilibrium. 
Set
\begin{equation}
a=|\delta_{14}\delta_{23}|^{1/3},\quad
b=|\delta_{12}\delta_{34}|^{1/3},\quad
c=|\delta_{13}\delta_{24}|^{1/3}.
\end{equation}
Then the following necessary conditions hold.

\begin{enumerate}
    \item[(i)] The three numbers \(a,b,c\) are the side lengths of a possibly
    degenerate triangle. Equivalently,
    \[
        a\leq b+c,\qquad b\leq c+a,\qquad c\leq a+b.
    \]

    \item[(ii)] If the triangle in \((i)\) is non-degenerate, let
    \(\theta_1,\theta_2,\theta_3\) be its interior angles opposite to
    \(a,b,c\), respectively, and let \(S_0\) be its area, $d$ be its circumdiameter. Equivalently, if this
    triangle is denoted by \(\triangle BDE\) with
    \[
        DE=a,\qquad BE=b,\qquad BD=c,
    \]
    then
    \[
        \theta_1=\angle DBE,\qquad
        \theta_2=\angle BDE,\qquad
        \theta_3=\angle BED.
    \]
    Define
    \[
        \mu_{12}=\Delta_{12}^{2/3}\Delta_{34}^{-1/3},\qquad
        \mu_{13}=\Delta_{13}^{2/3}\Delta_{24}^{-1/3},\qquad
        \mu_{14}=\Delta_{14}^{2/3}\Delta_{23}^{-1/3}.
    \]
    Then the following three homogeneous polynomial equations have a non-trivial
    common solution. Consequently,
    \[
        \operatorname{Res}(F_0,F_1,F_2)=0,
    \]
    where
    \begin{eqnarray}
    F_0(x,y,z)
    &=&
    \frac{1}{\mu_{14}^2}x^2
    \bigl(z^2+x^2+2zx\cos\theta_2\bigr)
    -
    \frac{1}{\mu_{12}^2}y^2
    \bigl(y^2+z^2+2yz\cos\theta_1\bigr),
    \label{F0:four-body-polynomials}\\
    F_1(x,y,z)
    &=&
    \frac{1}{\mu_{14}^2}x^2
    \bigl(x^2+y^2+2xy\cos\theta_3\bigr)
    -
    \frac{1}{\mu_{13}^2}z^2
    \bigl(y^2+z^2+2yz\cos\theta_1\bigr),
    \label{F1:four-body-polynomials}\\
    F_2(x,y,z)
    &=&
    \frac{\mu_{12}^2\mu_{13}^2}{4{d}^2S_0^2}
    x^2(ax+by+cz)^4
    -
    \bigl(y^2+z^2+2yz\cos\theta_1\bigr)
    (ayz+bzx+cxy)^2.
    \label{F2:four-body-polynomials}
    \end{eqnarray}
\end{enumerate}
In the degenerate case, namely when one of the equalities in \((i)\) holds, the
condition \((ii)\) should be replaced by the corresponding limiting or special
case. We do not use the formula involving \(S_0^{-2}\) in that case.
\end{theorem}

\begin{proof}
For convenience, set
\[
    \mu_{ij}
    =
    |\delta_{ij}|^{2/3}|\delta_{kl}|^{-1/3},
    \qquad
    \{i,j,k,l\}=\{1,2,3,4\}.
\]
Thus, for example,
\[
    \mu_{12}=|\delta_{12}|^{2/3}|\delta_{34}|^{-1/3}=\frac{|\delta_{12}|}{b},
    \quad
    \mu_{13}=|\delta_{13}|^{2/3}|\delta_{24}|^{-1/3}=\frac{|\delta_{13}|}{c},
    \quad
    \mu_{14}=|\delta_{14}|^{2/3}|\delta_{23}|^{-1/3}=\frac{|\delta_{14}|}{a}.
\]

Since \(A=a_1\), \(B=a_2\), \(C=a_3\), \(D=a_4\) form a
non-collinear equilibrium,
by Lemma~\ref{lem:geometric-characterization}, the six quantities
\[
    \frac{AB}{R_A R_B},\quad
    \frac{BC}{R_B R_C},\quad
    \frac{CD}{R_C R_D},\quad
    \frac{DA}{R_D R_A},\quad
    \frac{AC}{R_A R_C},\quad
    \frac{BD}{R_B R_D}
\]
are proportional to
\(\mu_{12},\;\mu_{23},\; \mu_{34},\;
    \mu_{14},\; \mu_{13},\; \mu_{24}\),
respectively.

We first derive the triangle condition. Construct a point \(E\) such that
\[
    \triangle ABE\sim \triangle ACD.
\]
Then one also has
\[
    \triangle ADE\sim \triangle ACB.
\]
Consequently,
\[
    \frac{DE}{BD}
    =
    \frac{AD\cdot BC}{AC\cdot BD},
    \qquad
    \frac{BE}{BD}
    =
    \frac{AB\cdot CD}{AC\cdot BD}.
\]
Using the proportional relation in Lemma~\ref{lem:geometric-characterization},
we obtain
\[
\begin{aligned}
\frac{DE}{BD}
&=
\frac{
\left(\dfrac{AD}{R_A R_D}\right)
\left(\dfrac{BC}{R_B R_C}\right)}
{
\left(\dfrac{AC}{R_A R_C}\right)
\left(\dfrac{BD}{R_B R_D}\right)}
=
\frac{\mu_{14}\mu_{23}}{\mu_{13}\mu_{24}}  \\
&=
\frac{
(\Delta_{14}\Delta_{23})^{1/3}}
{
(\Delta_{13}\Delta_{24})^{1/3}},
\end{aligned}
\]
and similarly
\[
\begin{aligned}
\frac{BE}{BD}
&=
\frac{
\left(\dfrac{AB}{R_A R_B}\right)
\left(\dfrac{CD}{R_C R_D}\right)}
{
\left(\dfrac{AC}{R_A R_C}\right)
\left(\dfrac{BD}{R_B R_D}\right)}
=
\frac{\mu_{12}\mu_{34}}{\mu_{13}\mu_{24}} \\
&=
\frac{
(\Delta_{12}\Delta_{34})^{1/3}}
{
(\Delta_{13}\Delta_{24})^{1/3}} .
\end{aligned}
\]
Therefore the three sides of the triangle \(BDE\) are proportional to
\[
    (\Delta_{14}\Delta_{23})^{1/3},\qquad
    (\Delta_{12}\Delta_{34})^{1/3},\qquad
    (\Delta_{13}\Delta_{24})^{1/3}.
\]
Equivalently, if we set
\[
    a=(\Delta_{14}\Delta_{23})^{1/3},\qquad
    b=(\Delta_{12}\Delta_{34})^{1/3},\qquad
    c=(\Delta_{13}\Delta_{24})^{1/3},
\]
then \(a,b,c\) must satisfy the triangle inequalities. This proves condition
\((i)\).

Now assume that the triangle determined by \(a,b,c\) is non-degenerate. We
normalize the auxiliary triangle \(BDE\) so that
\[
    DE=a,\qquad BE=b,\qquad BD=c.
\]
By assumption, we have
\[
    \theta_1=\angle DBE,\qquad
    \theta_2=\angle BDE,\qquad
    \theta_3=\angle BED,
\]
i.e., \(\theta_1,\theta_2,\theta_3\) are the angles opposite to
\(a,b,c\), respectively. Moreover,
\[
    S_0=S_{\triangle BDE}.
\]

Let \(B',D',E'\) be the feet of the perpendiculars from \(A\) to the three
sides \(DE,BE,BD\), respectively:
\[
    AB'\perp DE,\qquad AD'\perp BE,\qquad AE'\perp BD.
\]
Set
\[
    x=AB',\qquad y=AD',\qquad z=AE'.
\]
Since the configuration is non-collinear and no three particles are collinear,
we have
\[
    (x,y,z)\neq (0,0,0).
\]

By similarity properties of triangles:
\[
R_B = R_{ACD} = R_{ABE} \cdot \frac{AC}{AE},
\quad
R_D = R_{ABC} = R_{ADE} \cdot \frac{AB}{AE}.
\]
It then follows that
\[
\begin{aligned}
AB \cdot R_C R_D &= R_{ABD} R_{ADE} \cdot \frac{AB^2}{AE^2}, \\
AD \cdot R_B R_C &= R_{ABD} R_{ABE} \cdot \frac{AD^2}{AE^2}, \\
AC \cdot R_B R_D &= R_{ABE} R_{ADE} \cdot \frac{AB \cdot AD \cdot AC}{AE^2}.
\end{aligned}
\]
From the theorem,
\[
R_{ABD} R_{ADE} \cdot \frac{AB^2}{AE^2}
:
R_{ABD} R_{ABE} \cdot \frac{AD^2}{AE^2}
:
R_{ABE} R_{ADE} \cdot \frac{AB \cdot AD \cdot AC}{AE^2}
= \mu_{12} : \mu_{14} : \mu_{13}.
\]

Taking the ratio of the third term to the first term:
\[
\frac{R_{ABE} \cdot AD \cdot AC}{R_{ABD} \cdot AB \cdot AE}
= \frac{\mu_{13}}{\mu_{12}}.
\]

By triangle similarity $\displaystyle \frac{AB}{AE} = \frac{AC}{AD}$, we obtain
\[
\frac{R_{ABE} \cdot AD^2}{R_{ABD} \cdot AE^2}
= \frac{1}{\mu_{12}} : \frac{1}{\mu_{13}}.
\]

Comparing the third term to the second term similarly yields
\[
\frac{R_{ABD} \cdot AE^2}{R_{ADE} \cdot AB^2}
= \frac{1}{\mu_{13}} : \frac{1}{\mu_{14}}.
\]

Combining these results:
\begin{equation}
    R_{ABE} \cdot AD^2 : R_{ADE} \cdot AB^2 : R_{ABD} \cdot AE^2
= \frac{1}{\mu_{12}} : \frac{1}{\mu_{14}} : \frac{1}{\mu_{13}}.
\end{equation}

Recall the circumradius formula
\[
R_{ABE} = \frac{AB \cdot AE \cdot BE}{4 S_{\triangle ABE}},
\]
with identical expressions for other triangles.

Substitution gives the area-side ratio relation:
\begin{equation}
    \frac{AD \cdot BE}{S_{\triangle ABE}}
:
\frac{AB \cdot DE}{S_{\triangle ADE}}
:
\frac{BD \cdot AE}{S_{\triangle ABD}}
=
\frac{1}{\mu_{12}} : \frac{1}{\mu_{14}} : \frac{1}{\mu_{13}}.
\end{equation}

In the following proof, we regard $\triangle BDE$ as given, and investigate the existence of point $A$ satisfying the above equilibrium relations.
To simplify computation, we construct the \textit{orthic triangle} $\triangle B'D'E'$ of $A$ with respect to $\triangle BDE$, such that
\[
AB' \perp DE,\quad AD' \perp BE,\quad AE' \perp BD.
\]
We distinguish two geometric cases.

{\bf Case I.} If \(A\) lies inside \(\triangle BDE\), see Figure \ref{fig:placeholder2}.

\begin{figure}
    \centering
    \includegraphics[width=0.5\linewidth]{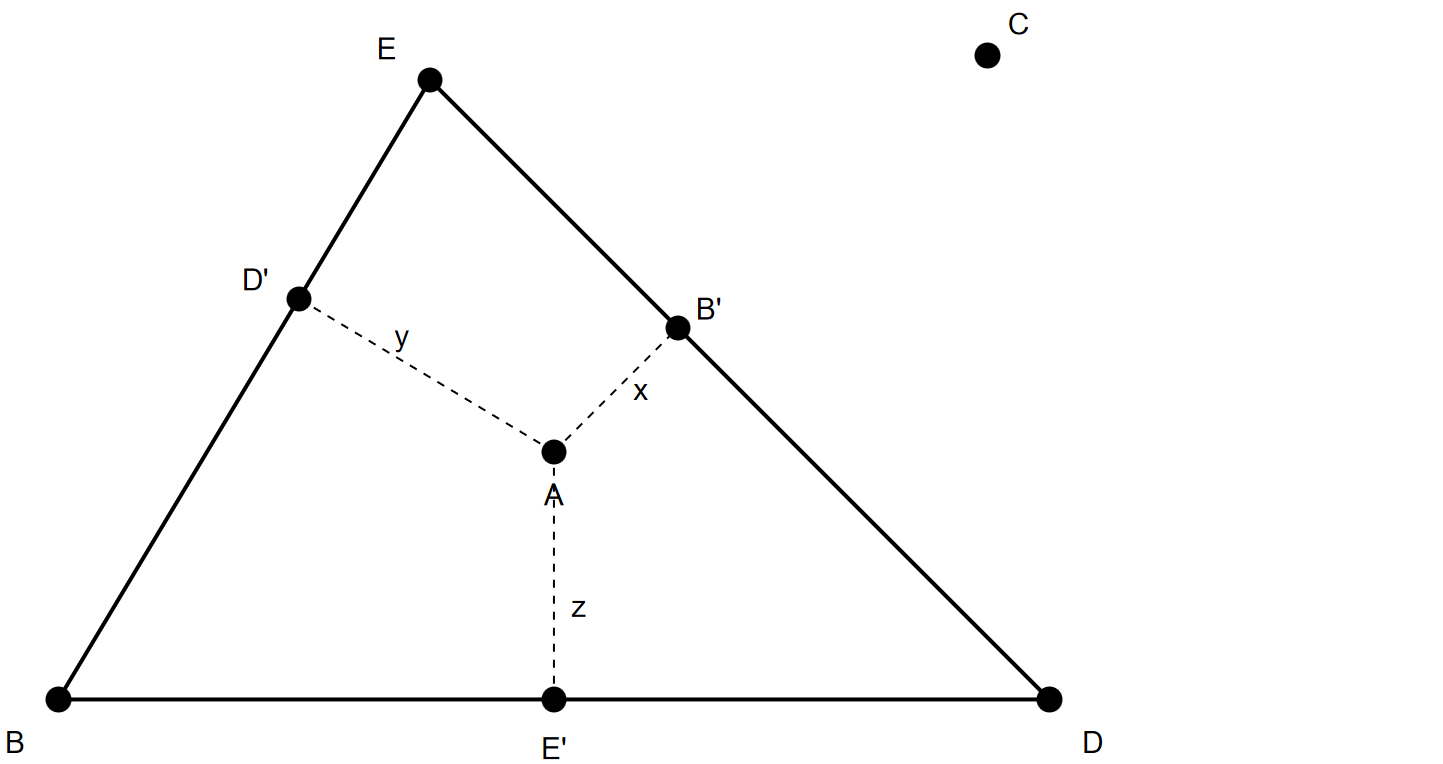}
    \caption{Point A lies inside triangle BDE}
    \label{fig:placeholder2}
\end{figure}

In such a case,
the angle chasing gives:
\[
\begin{aligned}
\angle D'B'E'
&= \angle AB'D + \angle AB'E
= \angle AED' + \angle ADE' \\
&= \angle ADC + \angle ADB
= \angle BDC, \\
\angle B'D'E'
&= \angle AED + \angle ABD
= \angle ABC + \angle ABD
= \angle CBD.
\end{aligned}
\]

Therefore
\[
\triangle B'D'E' \sim \triangle DBC.
\]

From similarity:
\[
\frac{R_A}{R_{B'D'E'}}
= \frac{BD}{B'D'}
= \frac{BD}{AE \sin E}
= \frac{2 R_{BDE}}{AE},
\]
so that
\[
R_A = \frac{2 R_{BDE}}{AE} \cdot R_{B'D'E'}.
\]

Additionally:
\[
R_C = R_{ABD} \cdot \frac{AE}{AE} = R_{ABD},
\quad
R_D = R_{ADE} \cdot \frac{AB}{AE}.
\]

Substituting into
\[
\frac{AB \cdot R_C R_D}{BC \cdot R_A R_D} = \frac{\mu_{12}}{\mu_{23}},
\]
we get
\[
\frac{R_{ABD} \cdot AB \cdot AE}{R_{B'D'E'} \cdot BC \cdot 2 R_{BDE}}
= \frac{\mu_{12}}{\mu_{23}}.
\]

Using $ \displaystyle BC = \frac{AB}{AE} \cdot DE $,
\[
\frac{R_{ABD} \cdot AE^2}{R_{B'D'E'} \cdot DE \cdot 2 R_{BDE}}
=
\frac{1}{\mu_{13}} : \left( \mu_{23} \cdot \frac{1}{\mu_{12}} \cdot \frac{1}{\mu_{13}} \right)
=\frac{1}{\mu_{13}} : \frac{\color{red}a}{\mu_{12}\mu_{14}\mu_{13}}.
\]

Collecting all four terms:
\[
\begin{aligned}
& R_{ABE} \cdot AD^2 : R_{ADE} \cdot AB^2 : R_{ABD} \cdot AE^2 : \left(2 R_{BDE} R_{B'D'E'}\right) \\
&=
\frac{1}{\mu_{12}} : \frac{1}{\mu_{14}} : \frac{1}{\mu_{13}} : \frac{1}{\mu_{12}\mu_{14}\mu_{13}}.
\end{aligned}
\]

Inserting the circumradius formula
\[
R_{B'D'E'} = \frac{B'D' \cdot D'E' \cdot E'B'}{4 S_{\triangle B'D'E'}}
= \frac{AE \cdot AB \cdot AD \cdot \sin B \sin D \sin E}{4 S_{\triangle B'D'E'}},
\]
we arrive at
\[
\begin{aligned}
&
\frac{AD \cdot BE}{4 S_{\triangle ABE}}
:
\frac{AB \cdot DE}{4 S_{\triangle ADE}}
:
\frac{AE \cdot BD}{4 S_{\triangle ABD}}
:
\frac{2 R_{BDE} \sin B \sin D \sin E}{4 S_{\triangle B'D'E'}} \\
&=
\frac{1}{\mu_{12}} : \frac{1}{\mu_{14}} : \frac{1}{\mu_{13}} : \frac{1}{\mu_{12}\mu_{14}\mu_{13}}.
\end{aligned}
\]
Equivalently:
\begin{equation}
    \begin{aligned}
&
\frac{AD}{2 AD'} : \frac{AB}{2 AB'} : \frac{AE}{2 AE'} : \frac{2 R_{BDE} \sin B \sin D \sin E}{4 S_{\triangle B'D'E'}} \\
&=
\frac{1}{\mu_{12}} : \frac{1}{\mu_{14}} : \frac{1}{\mu_{13}} : \frac{1}{\mu_{12}\mu_{14}\mu_{13}}.
\end{aligned}
\end{equation}
Finally, by the law of sines
\[
AD = \frac{B'E'}{\sin D} = B'E' \cdot \frac{2 R_{BDE}}{BE},
\]
and cyclic permutations, we obtain
\begin{equation}\label{proportional.relation}
\begin{aligned}
&
\frac{B'E'}{2 AD'} : \frac{D'E'}{2 AB'} : \frac{B'D'}{2 AE'} : \frac{\sin B \sin D \sin E}{4 S_{\triangle B'D'E'}} \\
&=
\frac{BE}{\mu_{12}} : \frac{DE}{\mu_{14}} : \frac{BD}{\mu_{13}} : \frac{1}{\mu_{12}\mu_{14}\mu_{13}} \\
&=
\mu_{34} : \mu_{23} : \mu_{24} : \frac{1}{\mu_{12}\mu_{14}\mu_{13}}.
\end{aligned}
\end{equation}
Hence, we have
\[
\frac{AD}{2AD'}: \frac{AB}{2AB'}: \frac{AE}{2AE'}
= \frac{1}{\mu_{14}}: \frac{1}{\mu_{12}}: \frac{1}{\mu_{13}},
\]

Suppose $k:=\frac{AD}{2AD'}:\frac{1}{\mu_{14}}$, which is a nonzero multiplier.
Recalling $AB' = x$, $AD' = y$, $AE' = z$, we obtain the system
\begin{eqnarray}
y^2+z^2+2yz \cos\theta_1 &=& \dfrac{4k^2}{\mu_{14}^2} x^2, \label{eq1}\\
z^2+x^2+2zx \cos\theta_2 &=& \dfrac{4k^2}{\mu_{12}^2} y^2, \label{eq2}\\
x^2+y^2+2xy \cos\theta_3 &=& \dfrac{4k^2}{\mu_{13}^2} z^2. \label{eq3}
\end{eqnarray}
Eliminating $k$ from \eqref{eq1} and \eqref{eq2} gives
\[
    F_0(x,y,z)=0,
\]
where
\[
F_0(x,y,z)
=
\frac{1}{\mu_{14}^2}x^2
\bigl(z^2+x^2+2zx\cos\theta_2\bigr)-
\frac{1}{\mu_{12}^2}y^2
\bigl(y^2+z^2+2yz\cos\theta_1\bigr).
\]
Similarly, eliminating $k$ from \eqref{eq1} and \eqref{eq3} gives
\[
    F_1(x,y,z)=0,
\]
where
\[
F_1(x,y,z)
=
\frac{1}{\mu_{14}^2}x^2
\bigl(x^2+y^2+2xy\cos\theta_3\bigr)  
-
\frac{1}{\mu_{13}^2}z^2
\bigl(y^2+z^2+2yz\cos\theta_1\bigr).
\]
Solving equations \eqref{eq1}-\eqref{eq3} yields the value of $k$. Since its explicit expression is complicated, we regard $k$ as a parameter in the following discussion.
By \eqref{proportional.relation}, we also have
\[
\frac{d \sin\theta_1 \sin\theta_2 \sin\theta_3}{4 S_{\triangle B'D'E'}}
= \frac{k}{\mu_{12}\mu_{14}\mu_{13}}.
\]

Now it remains to obtain the third polynomial equation. 
Recall that \(A\) lies inside \(\triangle BDE\), then the area relations give
\[
\begin{aligned}
2 S_{\triangle B'D'E'}
&= xy \sin\theta_3 + yz \sin\theta_1 + zx \sin\theta_2 \\
&= \frac{1}{d}\left(
|\delta_{13}|^{\frac{1}{3}} |\delta_{24}|^{\frac{1}{3}} xy
+ |\delta_{14}|^{\frac{1}{3}} |\delta_{23}|^{\frac{1}{3}} yz
+ |\delta_{14}|^{\frac{1}{3}} |\delta_{34}|^{\frac{1}{3}} zx
\right),
\end{aligned}
\]
and
\[
2 S_{\triangle BDE}
= |\delta_{14}|^{\frac{1}{3}} |\delta_{23}|^{\frac{1}{3}} x
+ |\delta_{12}|^{\frac{1}{3}} |\delta_{34}|^{\frac{1}{3}} y
+ |\delta_{13}|^{\frac{1}{3}} |\delta_{24}|^{\frac{1}{3}} z.
\]
Thus
\[
\frac{2 S_{\triangle BDE}}{4d\, S_{\triangle B'D'E'}}
= \frac{d^{2} \sin\theta_1 \cdot \sin\theta_2 \cdot \sin\theta_3}{4d\, S_{\triangle B'D'E'}}
= \frac{k{d}}{\mu_{12}\mu_{14}\mu_{13}}.
\]
Since
\(a = |\delta_{14} \delta_{23}|^{\frac{1}{3}},\;
b = |\delta_{12} \delta_{34}|^{\frac{1}{3}},\;
c = |\delta_{13} \delta_{24}|^{\frac{1}{3}}\),
 we get
\begin{equation}\label{area.eqs}
    \begin{cases}
ax + by + cz = 2 S_{\triangle BDE} \\[1ex]
ayz + bzx + cxy = \dfrac{\mu_{12}\mu_{14}\mu_{13}}{k{d}} S_{\triangle BDE}
\end{cases}
\end{equation}
Let $\mu = \mu_{12}\mu_{14}\mu_{13}$ and $S_0 = S_{\triangle BDE}$. It follows that
\[
ayz + bzx + cxy=\frac{\mu}{kd}S_0
= \frac{\mu}{k{d}} \cdot \frac{S_0}{4S_0^2} \left(ax + by + cz\right)^2,
\]
which gives a homogeneous equation with respect to $x,y,z$:
\begin{equation}\label{eq4}
(ax + by + cz)^2 - \frac{4k{d}S_0}{\mu} (ayz + bzx + cxy) = 0.
\end{equation}
Solving $k$ from \eqref{eq4}, we obtain
\begin{equation}\label{k}
    k=\frac{\mu(ax+by+cz)^2}{4{d}S_0(ayz+bzx+cxy)}.
\end{equation}
Plugging \eqref{k} into \eqref{eq1}, putting the terms over a common denominator, and using \(\mu=\mu_{12}\mu_{13}\mu_{14}\), we have
\[
    F_2^{+}(x,y,z)=0,
\]
where
\begin{equation}\label{F2}
F_2^{+}(x,y,z):=\frac{\mu_{12}^2\mu_{13}^2}{4{d}^2S_0^2}x^2(ax+by+cz)^4
    -(y^2+z^2+2yz\cos\theta_1)(ayz+bzx+cxy)^2.
\end{equation}

{\bf Case II.}  If \(A\) lies outside \(\triangle BDE\), see figure \ref{fig:placeholder3}.
Without loss of generality, we suppose $E$ lies interior of \(\triangle ABD\).

\begin{figure}
    \centering
    \includegraphics[width=0.5\linewidth]{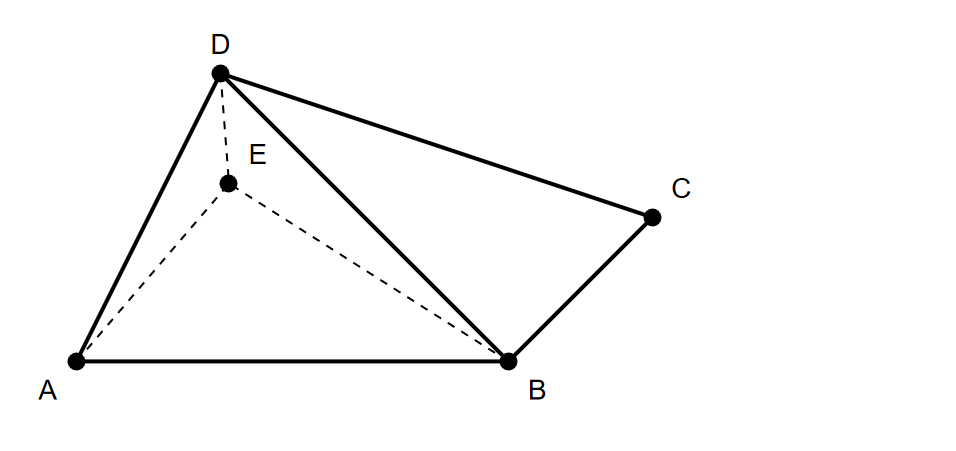}
    \caption{Point A lies outside triangle BDE}
    \label{fig:placeholder3}
\end{figure}

Under a similar arguments with Case I
Suppose $k:=\frac{AD}{2AD'}:\frac{1}{\mu_{14}}$, which is a nonzero multiplier.
Recalling $AB' = x$, $AD' = y$, $AE' = z$, 
under a similar arguments with Case I, we obtain the system
\begin{eqnarray}
y^2+z^2-2yz \cos\theta_1 &=& \dfrac{4k^2}{\mu_{14}^2} x^2, \label{eq1.outside}\\
z^2+x^2-2zx \cos\theta_2 &=& \dfrac{4k^2}{\mu_{12}^2} y^2, \label{eq2.outside}\\
x^2+y^2+2xy \cos\theta_3 &=& \dfrac{4k^2}{\mu_{13}^2} z^2. \label{eq3.outside}
\end{eqnarray}

Now the corresponding signed area
relations becomes
\[
    -ax-by+cz=2S_0
\]
and
\[
    ayz+bzx-cxy=\frac{\mu_{12}\mu_{13}\mu_{14}}{k{d}}S_0.
\]
The same substitution gives
\[
    F_2^{-}(x,y,z)=0,
\]
where
\begin{equation}\label{F2-}
F_2^{-}(x,y,z)
=
\frac{\mu_{12}^2\mu_{13}^2}{4{d}^2S_0^2}
x^2(-ax-by+cz)^4-
\bigl(y^2+z^2+2yz\cos\theta_1\bigr)
(ayz+bzx-cxy)^2 .
\end{equation}

Notes that, under the involution $(x,y,z)\to(-x,-y,z)$, the four polynomial
\eqref{eq1.outside}-\eqref{F2-} are just the same polynomials as
\eqref{eq1}-\eqref{F2}.

Thus, nevertheless whether $A$ in the interior or exterior of $\Delta BDE$,
the three homogeneous polynomials
\[
    F_0,\quad F_1,\quad F_2
\]
have a non-trivial common zero \((x,y,z)\). Hence, by the defining property of
the multivariate resultant,
\[
    \operatorname{Res}(F_0,F_1,F_2)=0.
\]
This proves condition \((ii)\) and completes the proof.
\end{proof}

\begin{remark}
Since $F_0,F_1,F_2$ have degrees $4,4,6$, respectively, 
it is difficult to compute their resultant in a direct manner.
However, we can equivalently establish the necessary condition by computing the resultants of two systems of three quadratic polynomials.

First, note that the system of three quadratic polynomial equations \eqref{eq1}-\eqref{eq3} has
a nontrivial solution $(x,y,z)\in\mathbb{R}^3$.
Then their resultant, which can be computed by Proposition \ref{prop:three-ternary-quadrics}, 
and is denoted by $R(k)$ as a function of the parameter $k$, must vanish.
Hence $k$ can be determined by solving the polynomial equation $R(k)=0$.

On the other hand,
expanding \eqref{eq4} yields:
\begin{equation}\label{eq4'}
    \begin{aligned}
a^2 x^2 + b^2 y^2 + c^2 z^2
&= -\left(2abxy + 2acxz + 2bcyz\right)
+ \frac{4k{d}S_0}{\mu}\left(ayz + bzx + cxy\right) \\[1ex]
&= \left(\frac{2k {d}S_0}{\mu}a - bc\right)\frac{1}{\cos B}\left[\frac{4k^2}{\mu_{14}^2}x^2 - y^2 - z^2\right] \\[1ex]
&\quad + \left(\frac{2k {d}S_0}{\mu}b - ac\right)\frac{1}{\cos D}\left[\frac{4k^2}{\mu_{12}^2}y^2 - x^2 - z^2\right] \\[1ex]
&\quad + \left(\frac{2k {d}S_0}{\mu}c - ab\right)\frac{1}{\cos E}\left[\frac{4k^2}{\mu_{13}^2}z^2 - x^2 - y^2\right].
\end{aligned}
\end{equation}
From \eqref{eq1}-\eqref{eq3} and \eqref{eq4'}, we obtain the system
\begin{equation}\label{three.quadratic.eqs}
    \begin{cases}
x^2 + b_1 xy + c_1 y^2 = 0 \\
y^2 + b_2 yz + c_2 z^2 = 0 \\
x^2 + b_3 zx + c_3 z^2 = 0
\end{cases}
\end{equation}

Let
\[
    P_1(t)=t^2+b_1t+c_1,\qquad
    P_2(t)=t^2+b_2t+c_2,\qquad
    P_3(t)=t^2+b_3t+c_3.
\]
Denote the roots of these three polynomials by
\[
    \xi_1,\xi_2;\qquad
    \eta_1,\eta_2;\qquad
    \zeta_1,\zeta_2,
\]
respectively. Namely,
$P_1(\xi_i)=0,\quad P_2(\eta_j)=0,\quad P_3(\zeta_\ell)=0$ for all $i,j,\ell=1,2$.

For any common solution \((x,y,z)\) of \eqref{three.quadratic.eqs} with \(xyz\ne0\), set
\[
    X=\frac{x}{y},\qquad
    Y=\frac{y}{z},\qquad
    Z=\frac{x}{z}.
\]
Then the three equations in \eqref{three.quadratic.eqs} are equivalent to
\[
    P_1(X)=0,\qquad P_2(Y)=0,\qquad P_3(Z)=0,
\]
and these three ratios satisfy
\[
    Z=XY.
\]
Conversely, if \(X,Y,Z\) satisfy the above three equations together with
\(Z=XY\), then choosing any \(z_0\ne0\) and setting
\[
    y=Yz_0,\qquad x=XYz_0
\]
gives a non-trivial common solution of \eqref{three.quadratic.eqs}.

Thus the problem is reduced to comparing the roots of \(P_3\) with all
pairwise products of roots of \(P_1\) and \(P_2\). Define the quartic
polynomial
\[
    Q(t)=\prod_{i=1}^2\prod_{j=1}^2 (t-\xi_i\eta_j).
\]
Its roots are
\[
    \xi_1\eta_1,\quad \xi_1\eta_2,\quad \xi_2\eta_1,\quad \xi_2\eta_2.
\]
Since
\[
    \xi_1+\xi_2=-b_1,\qquad \xi_1\xi_2=c_1,
    \qquad
    \eta_1+\eta_2=-b_2,\qquad \eta_1\eta_2=c_2,
\]
we obtain
\[
\begin{aligned}
    Q(t)
    ={}& t^4-b_1b_2t^3
    +(b_1^2c_2+b_2^2c_1-2c_1c_2)t^2   \\
    &{}-b_1b_2c_1c_2t+c_1^2c_2^2 .
\end{aligned}
\]
Therefore the quadratic system \eqref{three.quadratic.eqs} has a non-trivial solution with
\(xyz\ne0\) only if \(Q(t)\) and \(P_3(t)\) have a common root. Equivalently,
\[
    \operatorname{Res}(Q,P_3)=0.
\]
Here $\operatorname{Res}(Q,P_3)$ is computed directly by the determinant:
\begin{equation}
    \begin{aligned}
&\operatorname{Res}(Q,P_3)=\begin{vmatrix}
1 & -b_1b_2 & b_1^2c_2 + b_2^2c_1 - 2c_1c_2 & -b_1c_1b_2c_2 & c_1^2c_2^2& 0 \\
0 & 1 & -b_1b_2 & b_1^2c_2 + b_2^2c_1 - 2c_1c_2 & -b_1c_1b_2c_2 & c_1^2c_2^2 \\
1 & b_3 & c_3 & 0 & 0 & 0 \\
0 & 1 & b_3 & c_3 & 0 & 0 \\
0 & 0 & 1 & b_3 & c_3 & 0 \\
0 & 0 & 0 & 1 & b_3 & c_3
\end{vmatrix}\\
&= \frac{1}{\mu_{12}^8 \mu_{13}^8}
\Bigg(256 c_1^4 k^8 \mu_{12}^8 + 128 c_1^3 k^6 \mu_{12}^4 \mu_{13}^2
\bigg(
2 b_1 b_3 k^2 \mu_{12}^2
- \left(b_3^2 - 2 c_3\right)
\left(\mu_{12}^4 - 2 k^2 \mu_{13}^2\right)
\bigg) \\
&\quad+ c_3^2 \mu_{12}^4 \mu_{13}^4
\bigg(
16 b_1^4 k^4 \mu_{12}^4
+ 16 b_1^3 b_3 k^4 \mu_{12}^2 \mu_{13}^2 + 4 b_1 b_3 c_3 k^2 \mu_{12}^2 \mu_{13}^4
+ c_3^2 \mu_{12}^4 \mu_{13}^4 \\
&\quad\quad+ 4 b_1^2 k^2 \mu_{13}^2
\big(
b_3^2 \mu_{12}^4
- 2 c_3 \mu_{12}^4
+ 4 c_3 k^2 \mu_{13}^2
\big)
\bigg) \\
&\quad+ 8 c_1 c_3 k^2 \mu_{12}^2 \mu_{13}^4
\bigg(
8 b_1^3 b_3 k^4 \mu_{12}^4 - \left(b_3^2 - 2 c_3\right) c_3 \mu_{12}^2 \mu_{13}^2
\left(\mu_{12}^4 - 2 k^2 \mu_{13}^2\right) 
\\
&\quad\quad+ 2 b_1 b_3 k^2 \mu_{13}^2
\big(
b_3^2 \mu_{12}^4
- 5 c_3 \mu_{12}^4
+ 4 c_3 k^2 \mu_{13}^2
\big) + b_1^2
\big(
-8 c_3 k^2 \mu_{12}^6
+ 8 b_3^2 k^4 \mu_{12}^2 \mu_{13}^2
\big)
\bigg) \\
&\quad+ 16 c_1^2
\bigg(
4 b_1 b_3 k^6 \mu_{12}^2 \mu_{13}^4
\big(
b_3^2 \mu_{12}^4
- 5 c_3 \mu_{12}^4
+ 4 c_3 k^2 \mu_{13}^2
\big) + 4 b_1^2 k^6 \mu_{12}^4 \mu_{13}^2
\big(
b_3^2 \mu_{12}^4
- 2 c_3 \mu_{12}^4
+ 4 c_3 k^2 \mu_{13}^2
\big) 
\\
&\quad\quad+ k^4 \mu_{13}^4
\big(
b_3^4 \mu_{12}^8
- 4 b_3^2 c_3 \mu_{12}^8 
+ 2 c_3^2
\big(
3 \mu_{12}^8
- 8 k^2 \mu_{12}^4 \mu_{13}^2
+ 8 k^4 \mu_{13}^4
\big)
\big)
\bigg)
\Bigg)
\end{aligned}
\end{equation}
When this determinant vanishes, the quadratic system has a non-trivial
projective common root over the algebraic closure. In the present problem,
one must still determine whether such a root gives real positive distance
variables and satisfies the corresponding geometric branch conditions.
\end{remark}

Lastly, we turn to the case in which the triangle $\Delta BDE$ is degenerate. 
The similarity relations yield $\angle ABC+\angle ADC=\pi$; hence $A,B,C,D$ are concyclic.
We shall prove that this case is impossible.

\begin{corollary}\label{cor:degenerate-auxiliary-triangle}
Under the non-trivial consumption,
any equilibrium of four bodies cannot form a concyclic quadrilateral.
\end{corollary}
\begin{proof}
We will prove, none of the following equalities can occur:
\[
    a+b=c,\qquad b+c=a,\qquad c+a=b.
\]

Now we suppose $a+b=c$.
By the construction in the proof of Theorem~\ref{thm:charged-four-body-equilibrium},
the equality \(a+b=c\) means that the auxiliary triangle \(\triangle BDE\)
is degenerate and that \(E\) lies on the segment \(BD\). The similarity
relations then imply \(    \angle ABC+\angle ADC=\pi\).
Hence \(A,B,C,D\) are concyclic. Since the four points are distinct and no
three of them are collinear, this gives a convex quadrilateral. Therefore,
by Proposition~\ref{prop:four-body-sign-patterns}, the necessary sign pattern is
    \[
        \delta_{12}>0,\;
        \delta_{23}>0,\;
        \delta_{34}>0,\;
        \delta_{14}>0\;\;{\rm and}\;\;\delta_{13}<0,\;
        \delta_{24}<0.
    \]
Since \(\delta_{ij}=m_i m_j(1-\lambda_i\lambda_j)\),
the above sign pattern gives
\[
    \lambda_1\lambda_3>1,\qquad
    \lambda_2\lambda_4>1,
\]
and
\[
    \lambda_1\lambda_2<1,\quad
    \lambda_2\lambda_3<1,\quad
    \lambda_3\lambda_4<1,\quad
    \lambda_4\lambda_1<1.
\]
In particular, \(\lambda_1\) and \(\lambda_3\) have the same sign, and
\(\lambda_2\) and \(\lambda_4\) have the same sign.

These two signs cannot be the same. Indeed, if all \(\lambda_i\) had the same
sign, then
\[
    \lambda_1\lambda_2<1,\qquad
    \lambda_3\lambda_4<1
\]
would imply
\[
    (\lambda_1\lambda_3)(\lambda_2\lambda_4)
    =
    (\lambda_1\lambda_2)(\lambda_3\lambda_4)
    <1,
\]
contradicting
\(\lambda_1\lambda_3>1,\;
    \lambda_2\lambda_4>1\).

Thus \(\lambda_1,\lambda_3\) have one sign and \(\lambda_2,\lambda_4\) have
the opposite sign.
Set
\[
    \alpha=|\lambda_1|,\qquad
    \beta=|\lambda_2|,\qquad
    \gamma=|\lambda_3|,\qquad
    \eta=|\lambda_4|.
\]
Then
\(    \alpha\gamma>1,\;
    \beta\eta>1\).
Moreover,
\[
\begin{aligned}
a^3=|\delta_{14}\delta_{23}|
&=
m_1m_2m_3m_4(1+\alpha\eta)(1+\beta\gamma),\\
b^3=|\delta_{12}\delta_{34}|
&=
m_1m_2m_3m_4(1+\alpha\beta)(1+\gamma\eta),\\
c^3=|\delta_{13}\delta_{24}|
&=
m_1m_2m_3m_4(\alpha\gamma-1)(\beta\eta-1).
\end{aligned}
\]
Therefore
\[
\begin{aligned}
a^3+b^3-c^3
={}&
m_1m_2m_3m_4
\Big[
(1+\alpha\eta)(1+\beta\gamma)
+
(1+\alpha\beta)(1+\gamma\eta)  
-
(\alpha\gamma-1)(\beta\eta-1)
\Big]\\
={}&m_1m_2m_3m_4(1+\alpha\beta+\alpha\gamma+\alpha\eta
+\beta\gamma+\beta\eta+\gamma\eta+\alpha\beta\gamma\eta)\\
>{}&0.
\end{aligned}
\]
Since \(a,b>0\), we have
\[
    c^3<a^3+b^3<(a+b)^3.
\]
Thus
\[
    c<a+b,
\]
which contradicts the assumed equality \(a+b=c\). Therefore the degenerate
case \(a+b=c\) cannot occur.

The two remaining degenerate possibilities \(b+c=a\) and \(c+a=b\) are excluded
in the same way by relabelling the four particles. This proves the corollary.
\end{proof}

\subsection{Special symmetric reductions}

We now consider two special reductions of the four-body equilibrium conditions
under additional symmetry assumptions. Recall that
\[
    DE=a=|\delta_{14}\delta_{23}|^{1/3},\qquad
    BE=b=|\delta_{12}\delta_{34}|^{1/3},\qquad
    BD=c=|\delta_{13}\delta_{24}|^{1/3}.
\]

\paragraph{Equilateral auxiliary triangle.}
Assume that $\delta_{12}=\delta_{13}=\delta_{14}$ and $\delta_{34}=\delta_{24}=\delta_{23}$.
Set
\[
    p=\delta_{12}=\delta_{13}=\delta_{14},
    \qquad
    q=\delta_{23}=\delta_{24}=\delta_{34}.
\]
Then $a=b=c=|pq|^{1/3}$.
Thus the auxiliary triangle \(\triangle BDE\) is equilateral. Moreover,
\[
    \mu_{12}=\mu_{13}=\mu_{14}=:\nu,
    \qquad
    \mu=\mu_{12}\mu_{13}\mu_{14}=\nu^3,
\]
where $\nu=|p|^{2/3}|q|^{-1/3}$.
Since \(a=b=c\), we have $S_0=S_{\triangle BDE}=\frac{\sqrt{3}}{4}a^2$
and $d=\frac{2\sqrt3{}}{3}a$.

From \eqref{area.eqs},
we obtain
\[
    x+y+z=\frac{\sqrt{3}}{2}a,
    \qquad
    xy+yz+zx=\frac{\mu}{k{\color{red}d}}S_0
    =\frac{3}{8}\frac{\mu a}{k}.
\]
On the other hand,
\[
    \theta_1=\theta_2=\theta_3=\frac{\pi}{3}.
\]
Therefore the system \((4.13)\)--\((4.15)\) becomes
\[
\begin{cases}
    y^2+yz+z^2=\dfrac{4k^2}{\nu^2}x^2,\\[1.2ex]
    z^2+zx+x^2=\dfrac{4k^2}{\nu^2}y^2,\\[1.2ex]
    x^2+xy+y^2=\dfrac{4k^2}{\nu^2}z^2.
\end{cases}
\]
We claim that every positive solution satisfies
\[
    x=y=z.
\]
Indeed, if \(x,y,z\) are not all equal, then, after a permutation, we may
assume \(x<y\leq z\) or \(x\le y< z\). 
Then we have
\[
\dfrac{4k^2}{\mu_{14}^2}x^2=y^2+yz+z^2
>y^2+yz+z^2-(x+y+z)(z-x)
=y^2+xy+x^2=\dfrac{4k^2}{\mu_{13}^2}z^2
=\dfrac{4k^2}{\mu_{14}^2}z^2.
\]
Similarly, we have
\[
\dfrac{4k^2}{\mu_{14}^2}x^2\ge\dfrac{4k^2}{\mu_{14}^2}y^2
\ge\dfrac{4k^2}{\mu_{14}^2}z^2.
\]
Together with $x\le y\le z$, we must have $x\le y<0$ and $-x=|x|>|z|$,
which implies
\[
x+y+z\le x+z<-|z|+z\le0,
\]
which yields a contradiction. 
Hence $x=y=z$.
Using $x+y+z=\frac{\sqrt{3}}{2}a$,
we obtain $x=y=z=\frac{\sqrt{3}}{6}a$.
Substituting this into the three quadratic equations gives
$\frac{4k^2}{\nu^2}=3$, and hence
\[
    k=\frac{\sqrt{3}}{2}\nu.
\]
On the other hand, substituting \(x=y=z\) into $xy+yz+zx=\frac{3}{8}\frac{\mu a}{k}$
gives
\[
    k=\frac{3\mu}{2a}.
\]
Therefore $\frac{\sqrt{3}}{2}\nu=\frac{3\nu^3}{2a}$,
and hence $a=\sqrt{3}\nu^2$.

Since $a=|pq|^{1/3},
    \;
    \nu^2=|p|^{4/3}|q|^{-2/3}$,
the identity \(a=\sqrt{3}\nu^2\) is equivalent to
\[
    |q|=\sqrt{3}|p|.
\]
It remains to impose the sign condition. 
In this reduction, since $x=y=z>0$, the particle \(A\)
lies inside the triangle formed by \(B,D,E\). Thus the three coefficients on
the outer triangle must have the same sign, while the three coefficients
joining \(A\) to the outer vertices must have the opposite sign. Therefore
\(p\) and \(q\) must have opposite signs. Combining this with
\(|q|=\sqrt{3}|p|\), we obtain
\[
    q=-\sqrt{3}p.
\]

The preceding computation shows that this reduction is very rigid. In fact,
a nonzero solution forces $x=y=z$.
Thus \(A\) is the center of the equilateral auxiliary triangle
\(\triangle BDE\). By the similarity relation $\triangle ABE\sim \triangle ACD$,
the auxiliary point \(E\) coincides with the original point \(C\), up to the
chosen orientation. Hence the original four-body configuration consists of an
equilateral triangle \(BCD\) together with its center \(A\). We summarize this
case as follows.
\begin{proposition}[Equilateral triangle with its center]
Assume
\[
    \delta_{12}=\delta_{13}=\delta_{14}=p,
    \qquad
    \delta_{23}=\delta_{24}=\delta_{34}=q,
\]
with \(pq\neq 0\). Then, in this symmetric reduction, a non-collinear
equilibrium exists if and only if
\[
    q=-\sqrt{3}p.
\]
Equivalently,
\[
    \delta_{23}=\delta_{24}=\delta_{34}
    =
    -\sqrt{3}\delta_{12}
    =
    -\sqrt{3}\delta_{13}
    =
    -\sqrt{3}\delta_{14}.
\]

In this case the original configuration has the following shape: the three
points \(B,C,D\) form an equilateral triangle, and \(A\) is its center. Thus the
four-body configuration is a concave configuration consisting of an equilateral
triangle together with its center.

Moreover, up to similarity, the corresponding solution of the reduced geometric
equations is
\[
    x=y=z=\frac{\sqrt{3}}{6}a,
    \qquad
    k=\frac{\sqrt{3}}{2}\nu
    =
    \frac{3\mu}{2a},
\]
where $a=|pq|^{1/3},\;\nu=|p|^{2/3}|q|^{-1/3},\;\mu=\nu^3$.
\end{proposition}

\medskip
\paragraph{An isosceles reduction.}
Assume that $\delta_{12}=\delta_{23}=\delta_{34}=\delta_{14}$.
Set
\[
    s=\delta_{12}=\delta_{23}=\delta_{34}=\delta_{14}.
\]
Then $a=|\delta_{14}\delta_{23}|^{1/3}=|\delta_{12}\delta_{34}|^{1/3}=|s|^{2/3}$
and $\mu_{12}=\mu_{23}=\mu_{34}=\mu_{41}=|s|^{1/3}$.
Thus $BE=DE=a$.
Moreover,
\[
    c=|\delta_{13}\delta_{24}|^{1/3}.
\]
We assume
\[
    0<c<2a,
\]
so that the auxiliary triangle \(\triangle BDE\) is non-degenerate. 
Since $\theta_1=\theta_2$, we have
\[
    \cos\theta_1=\sin\frac{\theta_3}{2},
    \qquad
    \cos\frac{\theta_3}{2}>0.
\]

We next show that any positive solution must satisfy \(x=y\). 
Subtracting the second equation in \((4.13)\)--\((4.15)\) from the first one,
we obtain
\[
    (1+\lambda)(y^2-x^2)+2z\cos\theta_1(y-x)=0,
\]
where $\lambda=\frac{4k^2}{\mu_{12}^2}=\frac{4k^2}{\mu_{14}^2}>0$.
Hence
\[
    (y-x)\bigl((1+\lambda)(x+y)+2z\cos\theta_1\bigr)=0.
\]
Since \(x,y,z>0\) and \(0<\theta_1<\pi/2\), the second factor is positive.
Therefore $x=y$.

From \((4.13)\)--\((4.15)\), we obtain
\[
\begin{cases}
    2xz\cos\theta_1
    =
    \left(\dfrac{4k^2}{\mu_{14}^2}-1\right)x^2-z^2,
    \\[1.5ex]
    2x^2\cos\theta_3
    =
    -2x^2+\dfrac{4k^2}{\mu_{13}^2}z^2.
\end{cases}
\]
The second equation gives
\[
    \frac{4k^2}{\mu_{13}^2}z^2
    =
    2x^2(1+\cos\theta_3)
    =
    4x^2\cos^2\frac{\theta_3}{2}.
\]
Hence
\begin{equation}\label{z.sym.case.2}
    z=\frac{\mu_{13}\cos(\theta_3/2)}{k}x.
\end{equation}
Substituting this into the first equation gives
\[
    \frac{4}{\mu_{14}^2}k^4
    -
    k^2
    -
    2\mu_{13}\cos\frac{\theta_3}{2}\cos\theta_1\, k
    -
    \mu_{13}^2\cos^2\frac{\theta_3}{2}
    =
    0.
\]
Next, the
signed area relations must also be satisfied. To express the inside and outside
branches simultaneously, let
\[
    \varepsilon\in\{1,-1\}.
\]
Here \(\varepsilon=1\) corresponds to the inside branch, while
\(\varepsilon=-1\) corresponds to the outside branch. The signed area relations
can be written as
\[
    2\varepsilon ax+cz=2S_0,
\]
and
\[
    2axz+\varepsilon c x^2=\frac{\mu}{k{d}}S_0.
\]
Using \eqref{z.sym.case.2},
the first signed area relation gives
\[
    x=
    \frac{2S_0 k}
    {2\varepsilon ak+c\mu_{13}\cos(\theta_3/2)}.
\]
The second signed area relation gives
\[
    x^2=
    \frac{\mu S_0}
    {{d}(2a\mu_{13}\cos(\theta_3/2)+\varepsilon ck)}.
\]
Therefore these two expressions are compatible if and only if
\[
    4S_0{d}k^2
    \left(
        2a\mu_{13}\cos\frac{\theta_3}{2}+\varepsilon ck
    \right)
    =
    \mu
    \left(
        2\varepsilon ak+c\mu_{13}\cos\frac{\theta_3}{2}
    \right)^2.
\]

In the isosceles reduction, the equalities $BE=DE,\; BA=DA$
imply that both \(A\) and \(E\) lie on the perpendicular bisector of \(BD\).
Using the similarity relations $\triangle ABE\sim \triangle ACD$
and $\triangle ADE\sim \triangle ACB$,
we see that the original point \(C\) also lies on this same symmetry axis.
Consequently,
\[
    AB=AD,\qquad CB=CD,
\]
and the original quadrilateral \(ABCD\) is a kite. 
The existence condition is
therefore reduced to the following algebraic and positivity conditions.

\begin{proposition}[Kite-shaped reduction]
Assume
\[
    \delta_{12}=\delta_{23}=\delta_{34}=\delta_{14}=s,
\]
and set
\[
    a=|s|^{2/3},
    \qquad
    c=|\delta_{13}\delta_{24}|^{1/3}.
\]
Suppose $0<c<2a$.
Then any non-collinear equilibrium in this reduction has the shape of a kite:
\[
    AB=AD,\qquad CB=CD.
\]
Equivalently, \(AC\) is the symmetry axis of the configuration, and the two
points \(B,D\) are symmetric with respect to \(AC\).

In this kite-shaped reduction, a non-collinear equilibrium exists if and only
if the following conditions hold:
\[
    s>0,
    \qquad
    \delta_{13}<0,
    \qquad
    \delta_{24}<0,
\]
and there exist $\varepsilon\in\{1,-1\}$ and $k>0$,
such that
\[
    \frac{4}{\mu_{14}^2}k^4
    -
    k^2
    -
    2\mu_{13}\cos\frac{\theta_3}{2}\cos\theta_1\, k
    -
    \mu_{13}^2\cos^2\frac{\theta_3}{2}
    =
    0,
\]
\[
    4S_0{d}k^2
    \left(
        2a\mu_{13}\cos\frac{\theta_3}{2}+\varepsilon ck
    \right)
    =
    \mu
    \left(
        2\varepsilon ak+c\mu_{13}\cos\frac{\theta_3}{2}
    \right)^2,
\]
and
\[
    2\varepsilon ak+c\mu_{13}\cos\frac{\theta_3}{2}>0,
    \qquad
    2a\mu_{13}\cos\frac{\theta_3}{2}+\varepsilon ck>0.
\]

In this case,
\[
    x=y=
    \frac{2S_0 k}
    {2\varepsilon ak+c\mu_{13}\cos(\theta_3/2)},
    \qquad
    z=
    \frac{\mu_{13}\cos(\theta_3/2)}{k}x.
\]
Here $\mu=\mu_{12}\mu_{13}\mu_{14}=a\mu_{13}$,
and the angles of the auxiliary triangle \(\triangle BDE\) satisfy
$\theta_1=\theta_2,
    \;
    \theta_3=\pi-2\theta_1$.
\end{proposition}

\section{The collinear equilibrium}

We now study collinear equilibria for a general charged \(N\)-body system. 
Compared with the non-collinear four-body case, the collinear setting allows us
to reduce the force-balance equations to a system of algebraic equations in the
\(N-1\) adjacent distances, which provides a natural way to derive necessary
conditions for the existence of equilibria.

\subsection{The necessary condition}

Given positive mass $m = (m_1, m_2,\ldots, m_N)\in(\mathbb{R}^+)^N$ and the quantities of
charges $e = (e_i,e_2,\ldots, e_N) \in\mathbb{R}^N$, 
let $a=(a_1,\ldots,a_N)\in\mathbb{R}^N$ be an $N$-body collinear equilibrium of $m$ for $1\le i\le N$
which satisfies $a_i<a_j$ if $i<j$.
The total force on mass $m_i,1\le i\le N$ is given by
\begin{equation}
    F_i=-\sum_{j=1}^{i-1}\frac{m_im_j-e_ie_j}{r_{ij}^2}
    +\sum_{j=i+1}^{N}\frac{m_im_j-e_ie_j}{r_{ij}^2},
\end{equation}
where $r_{ij}=|a_i-a_j|$.

Suppose $x_i=a_{i+1}-a_i$ for $1\le i\le N-1$.
Thus, an equilibrium exists if and only if the equations
$F_1(x_1,\ldots,x_{N-1})=\ldots=F_n(x_1,\ldots,x_{N-1})=0$
possess a positive solution $(x_1,\ldots,x_{N-1})\in(\mathbb{R}^+)^{N-1}$.
Moreover, since $F_1+F_2+\ldots+F_N\equiv0$, we only need to consider the first $N-1$ equations.
To eliminate the denominators of $F_i$, we define
\begin{equation}\label{eq:collinear-Gi}
    G_i(x_1,\ldots,x_{N-1})=(\Pi_{j\not=i}{r_{ij}^2})F_i.
\end{equation}
Now $G_i(x_1,\ldots,x_{N-1})$ is a homogeneous multivariable polynomial in $x_1,\ldots,x_{N-1}$.

Then by Theorem 2.3 in Chapter 3 of \cite{CLO}, we have
the following result.
\begin{theorem}\label{thm:necessary.condition.collinear.nbp}
    Consider a charged $N$-body system as in Definition \ref{def:e}. For each value of the masses $m = (m_1, m_2, \ldots,m_N)\in(\mathbb{R}^+)^N$ and the quantities of
charges $e = (e_1,e_2,\ldots, e_N) \in\mathbb{R}^N$,
    suppose at least two of the parameters $\lambda_{ij},1\le i<j\le N$ are nonzero.
Then a necessary condition for the existence of an equilibrium is the resultant of the $N-1$ multivariable
polynomial $G_1(x_1,\ldots,x_{N-1}),G_2(x_1,\ldots,x_{N-1}),\ldots,G_{N-1}(x_1,\ldots,x_{N-1})$ is zero,
i.e.,
\begin{equation}
    Res(G_1,G_2,\ldots,G_{N-1})=0.
\end{equation}
\end{theorem}

\begin{remark}
Although Theorem~\ref{thm:necessary.condition.collinear.nbp} is formulated only for collinear equilibria, the same idea
can in principle be applied to general non-collision equilibria. After choosing
coordinate variables for the configuration and clearing all denominators in the
force-balance equations, the equilibrium condition can be transformed into a
system of homogeneous multivariable polynomial equations. Therefore the
existence of such an equilibrium implies the vanishing of the corresponding
multivariate resultant.

However, in the general non-collinear case, the resulting polynomials and their
resultant are usually too complicated to be useful in explicit form. We therefore
restrict the explicit resultant formulation in this section to collinear
equilibria.
\end{remark}

Theorem~\ref{thm:necessary.condition.collinear.nbp} gives a general resultant necessary condition for collinear equilibria. 
In practice, however, the resultant becomes difficult to compute and
to interpret as \(N\) increases. We therefore spell out the first two nontrivial
cases. The case \(N=3\) recovers the square-root relation obtained in
Theorem~\ref{thm:charged-three-body-equilibrium} from the viewpoint of resultants, while the case \(N=4\) already
illustrates the rapid growth of the algebraic complexity and motivates the
special reductions considered below.

{\bf Example $N=3$.}

In such a case, we have
\begin{eqnarray}
    G_1(x_1,x_2)&=&x_1^2(x_1+x_2)^2\left(\frac{m_1m_2-e_1e_2}{x_1^2}+\frac{m_1m_3-e_1e_3}{(x_1+x_2)^2}\right)
    \nonumber\\
    &=&(\delta_{12}+\delta_{13})x_1^2+2\delta_{12}x_1x_2+\delta_{12}x_2^2,
    \\
    G_2(x_1,x_2)&=&x_1^2x_2^2\left(-\frac{m_1m_2-e_1e_2}{x_1^2}+\frac{m_2m_3-e_2e_3}{x_2^2}\right)
    \nonumber\\
    &=&\delta_{23}x_1^2-\delta_{12}x_2^2.
\end{eqnarray}
Then we obtain
\begin{eqnarray}
    Res(G_1,G_2)
    &=&\left|\begin{matrix}
        \delta_{12}& 2\delta_{12}& \delta_{12}+\delta_{13}& 0\\
        0& \delta_{12}& 2\delta_{12}& \delta_{12}+\delta_{13}\\
        -\delta_{12}& 0& \delta_{23}& 0\\
        0& -\delta_{12}& 0& \delta_{23}
    \end{matrix}\right|
    \nonumber\\
    &=&\delta_{12}^2(\delta_{12}^2+\delta_{13}^2+\delta_{23}^2+2\delta_{12}\delta_{13}+2\delta_{13}\delta_{23}-2\delta_{12}\delta_{23})
    \nonumber\\
    &=&\delta_{12}^2(\delta_{12}+\delta_{13}+\delta_{23}+2\sqrt{\delta_{12}\delta_{23}})(\delta_{12}+\delta_{13}+\delta_{23}-2\sqrt{\delta_{12}\delta_{23}}).
\end{eqnarray}
Therefore, $Res(G_1,G_2)=0$ leads to three subcases:

(i) $\delta_{12}=0$.
In this subcase, $G_1$ and $G_2$ have no nonzero common solutions.
Hence, an equilibrium is impossible.
In the remaining two cases, we suppose $\delta_{12}\ne0$.

(ii)
$\delta_{12}+\delta_{13}+\delta_{23}+2\sqrt{\delta_{12}\delta_{23}}=0$.
In this subcase, $\delta_{12}$ and $\delta_{23}$ have the same sign.
If $\delta_{12}>0$, we must have $\delta_{13}<0$ and 
$$
\sqrt{-\delta_{13}}=\sqrt{\delta_{12}}+\sqrt{\delta_{23}},
$$
which coincides with \eqref{condition.of.e.1}.
If $\delta_{12}<0$, it follows that
$$
\delta_{13}=(\sqrt{-\delta_{12}}-\sqrt{-\delta_{23}})^2,
$$
and therefore $\delta_{13}>0$.
Direct computation shows that $G_1(x_1,x_2)$ and $G_2(x_1,x_2)$ have only solutions of the form $x_1=x\sqrt{-\delta_{12}},x_2=-x\sqrt{-\delta_{23}}$ for any $x\in\mathbb{R}$.
These two variables cannot both be positive.
Hence, there is no equilibrium when $\delta_{12}<0$ in this subcase.

(iii) $\delta_{12}+\delta_{13}+\delta_{23}-2\sqrt{\delta_{12}\delta_{23}}=0$.
By an argument similar to that in subcase (ii), an equilibrium exists if and only if $\delta_{12},\delta_{23}<0,\delta_{13}>0$, and the following holds:
$$
\sqrt{\delta_{13}}=\sqrt{-\delta_{12}}+\sqrt{-\delta_{23}}.
$$
This condition coincides with \eqref{condition.of.e.2}.

\medskip
{\bf Example $N=4$.}

In such a case, we have
\begin{eqnarray}
    G_1(x_1,x_2,x_3)&=&x_1^2(x_1+x_2)^2(x_1+x_2+x_3)^2\left(\frac{m_1m_2-e_1e_2}{x_1^2}+\frac{m_1m_3-e_1e_3}{(x_1+x_2)^2}+\frac{m_1m_4-e_1e_4}{(x_1+x_2+x_3)^2}\right)
    \nonumber\\
    &=&(\delta_{12}+\delta_{13}+\delta_{14})x_1^4
+(4\delta_{12}+2\delta_{13}+2\delta_{14})x_1^3x_2
+(2\delta_{12}+2\delta_{13})x_1^3x_3  \\
&&+(6\delta_{12}+\delta_{13}+\delta_{14})x_1^2x_2^2
+(6\delta_{12}+2\delta_{13})x_1^2x_2x_3
+(\delta_{12}+\delta_{13})x_1^2x_3^2 \\
&&+4\delta_{12}x_1x_2^3
+6\delta_{12}x_1x_2^2x_3
+2\delta_{12}x_1x_2x_3^2
+\delta_{12}x_2^4
+2\delta_{12}x_2^3x_3
+\delta_{12}x_2^2x_3^2,\label{G1}
    \\
    G_2(x_1,x_2,x_3)&=&x_1^2x_2^2(x_2+x_3)^2\left(-\frac{m_1m_2-e_1e_2}{x_1^2}+\frac{m_2m_3-e_2e_3}{x_2^2}+\frac{m_2m_4-e_2e_4}{(x_2+x_3)^2}\right)
    \nonumber\\
    &=&(\delta_{23}+\delta_{24})x_1^2x_2^2+2\delta_{23}x_1^2x_2x_3+\delta_{23}x_1^2x_3^2-2\delta_{12}x_2^4-2\delta_{12}x_2^3x_3-\delta_{12}x_2^2x_3^2,
    \\
    G_3(x_1,x_2,x_3)&=&(x_1+x_2)^2x_2^2x_3^2\left(-\frac{m_1m_3-e_1e_3}{(x_1+x_2)^2}-\frac{m_2m_3-e_2e_3}{x_2^2}+\frac{m_3m_4-e_3e_4}{x_3^2}\right)
    \nonumber\\
    &=&\delta_{34}x_1^2x_2^2-\delta_{23}x_1^2x_3^2+2\delta_{34}x_1x_2^3-2\delta_{23}x_1x_2x_3^2
    +\delta_{34}x_2^4-(\delta_{13}+\delta_{23})x_2^2x_3^2.\label{G3}
\end{eqnarray}

Computing the resultant of the above three ternary quartic polynomials, given by (\ref{G1})-(\ref{G3}), still exceeds the computational capacity of an ordinary personal computer, and the final result is extremely complicated — so much so that if written out in full, it would occupy several pages. In order to present it in the paper, we have made some simplifications: we assume that 
\begin{equation}
    \delta_{14}=\delta_{23}=0.
\end{equation}
Under these constraints, we have
\begin{eqnarray}
    G_1(x_1,x_2,x_3)&=&(x_1+x_2+x_3)^2\tilde{G}_1(x_1,x_2,x_3),
    \nonumber\\
    G_2(x_1,x_2,x_3)&=&x_2^2\tilde{G}_2(x_1,x_2,x_3),
    \nonumber\\
    G_3(x_1,x_2,x_3)&=&x_2^2\tilde{G}_3(x_1,x_2,x_3),\nonumber
\end{eqnarray}
where
\begin{eqnarray}
    \widetilde{G}_1(x_1,x_2,x_3)&=&(\delta_{12}+\delta_{13})x_1^2+2\delta_{12}x_1x_2
    +\delta_{12}x_2^2,
    \label{tilde.G1}\\
    \widetilde{G}_2(x_1,x_2,x_3)&=&\delta_{24}x_1^2-\delta_{12}x_2^2-2\delta_{12}x_2x_3-\delta_{12}x_3^2,
    \label{tilde.G2}\\
    \widetilde{G}_3(x_1,x_2,x_3)&=&\delta_{34}x_1^2+2\delta_{34}x_1x_2+\delta_{34}x_2^2-\delta_{13}x_3^2.\label{tilde.G3}
\end{eqnarray}
Since $(x_1,x_2,x_{3})\in(\mathbb{R}^+)^{3}$, 
by Theorem \ref{thm:resultant-properties},
it is reasonable to compute
$Res(\widetilde{G}_1,\widetilde{G}_2,\widetilde{G}_3)$.
Using Maple, we obtain
\[
\operatorname{Res}(\widetilde G_1,\widetilde G_2,\widetilde G_3)
=
\delta_{12}^{8}\delta_{13}^{8}(\delta_{12}+\delta_{13})^{4}
P_4,
\]
where
\[
\begin{aligned}
P_4={}&
\delta_{12}^{4}+4\delta_{13}\delta_{12}^{3}
-4\delta_{24}\delta_{12}^{3}+4\delta_{34}\delta_{12}^{3}                         
+6\delta_{13}^{2}\delta_{12}^{2}+6\delta_{24}^{2}\delta_{12}^{2}
+6\delta_{34}^{2}\delta_{12}^{2}-4\delta_{13}\delta_{24}\delta_{12}^{2}  
+4\delta_{13}\delta_{34}\delta_{12}^{2}\\
&-4\delta_{24}\delta_{34}\delta_{12}^{2}
+4\delta_{13}^{3}\delta_{12}-4\delta_{24}^{3}\delta_{12}                         +4\delta_{34}^{3}\delta_{12}-4\delta_{13}\delta_{24}^{2}\delta_{12}
-4\delta_{13}\delta_{34}^{2}\delta_{12}+4\delta_{24}\delta_{34}^{2}\delta_{12} 
+4\delta_{13}^{2}\delta_{24}\delta_{12}
\\
&-4\delta_{13}^{2}\delta_{34}\delta_{12}
-4\delta_{24}^{2}\delta_{34}\delta_{12}-40\delta_{13}\delta_{24}\delta_{34}\delta_{12}                      
+\delta_{13}^{4}+\delta_{24}^{4}
+\delta_{34}^{4}+4\delta_{13}\delta_{24}^{3}
-4\delta_{13}\delta_{34}^{3}+4\delta_{24}\delta_{34}^{3}                                        \\
&+6\delta_{13}^{2}\delta_{24}^{2}+6\delta_{13}^{2}\delta_{34}^{2}
+6\delta_{24}^{2}\delta_{34}^{2}-4\delta_{13}\delta_{24}\delta_{34}^{2}          +4\delta_{13}^{3}\delta_{24}-4\delta_{13}^{3}\delta_{34}
+4\delta_{24}^{3}\delta_{34}+4\delta_{13}\delta_{24}^{2}\delta_{34}
-4\delta_{13}^{2}\delta_{24}\delta_{34}
\\
=&\prod_{\epsilon,\eta=\pm1}(\delta_{12}+\delta_{13}+\delta_{24}+\delta_{34}+2\epsilon\sqrt{\delta_{12}\delta_{24}}+2\eta\sqrt{\delta_{13}\delta_{34}}).
\end{aligned}
\]

Therefore, $Res(\widetilde{G}_1,\widetilde{G}_2,\widetilde{G}_3)=0$ leads to three subcases:

(i) $\delta_{12}=0$.
In such a case $\widetilde G_1=\delta_{13}x_1^2$.
Hence, unless further degeneracies occur, there is no positive solution. 
Similarly, if \(\delta_{13}=0\), then $\widetilde G_3=\delta_{34}(x_1+x_2)^2$,
which again gives no positive solution in the non-degenerate case.

(ii) $\delta_{12}+\delta_{13}=0$ and \(\delta_{12}\neq 0\).
In such a case, we have
\[
    \widetilde G_1
    =
    \delta_{12}\bigl((x_1+x_2)^2-x_1^2\bigr)
    =
    \delta_{12}x_2(2x_1+x_2).
\]
Since \(x_1,x_2>0\), this cannot vanish. Therefore the factor
\((\delta_{12}+\delta_{13})^4\) does not produce a positive collinear
equilibrium.

(iii) $P_4=0$ and \(\delta_{12}\neq 0,\delta_{12}+\delta_{13}\ne0\).
In such a case,
if a positive solution of equations \eqref{tilde.G1}-\eqref{tilde.G3} exists, then from
\[
    \delta_{12}(x_1+x_2)^2+\delta_{13}x_1^2=0
\]
we see that \(\delta_{12}\) and \(\delta_{13}\) must have opposite signs.
Moreover,
\[
    \frac{x_1+x_2}{x_1}
    =
    \sqrt{-\frac{\delta_{13}}{\delta_{12}}}
    >1.
\]
The equation  \eqref{tilde.G2} implies that \(\delta_{24}\) and \(\delta_{12}\) have the
same sign, and the equation \eqref{tilde.G3} implies that \(\delta_{34}\) and
\(\delta_{13}\) have the same sign. Hence a positive solution can exist only in
one of the following two sign patterns:
\[
    \delta_{12}>0,\quad \delta_{24}>0,\quad
    \delta_{13}<0,\quad \delta_{34}<0,
\]
or
\[
    \delta_{12}<0,\quad \delta_{24}<0,\quad
    \delta_{13}>0,\quad \delta_{34}>0.
\]

Assume first that
\[
    \delta_{12}>0,\qquad
    \delta_{24}>0,\qquad
    \delta_{13}<0,\qquad
    \delta_{34}<0.
\]
Then the three equations \eqref{tilde.G1}-\eqref{tilde.G3} are equivalent to
\begin{equation}\label{lin.eq.of.xs}
    x_1+x_2
    =
    \frac{\sqrt{-\delta_{13}}}{\sqrt{\delta_{12}}}\,x_1,\quad
    x_2+x_3
    =
    \frac{\sqrt{\delta_{24}}}{\sqrt{\delta_{12}}}\,x_1,\quad
    x_3
    =
    \frac{\sqrt{-\delta_{34}}}{\sqrt{\delta_{12}}}\,x_1.
\end{equation}
Therefore
\[
    x_2
    =
    \frac{\sqrt{-\delta_{13}}-\sqrt{\delta_{12}}}
         {\sqrt{\delta_{12}}}\,x_1.
\]
The condition \(x_2>0\) is equivalent to
\[
    \sqrt{-\delta_{13}}>\sqrt{\delta_{12}}.
\]
Substituting the expressions for \(x_2\) and \(x_3\) into the second equation
in \eqref{lin.eq.of.xs}
gives
\[
    \sqrt{\delta_{24}}
    =
    \sqrt{-\delta_{13}}-\sqrt{\delta_{12}}
    +\sqrt{-\delta_{34}},
\]
or equivalently
\[
    \sqrt{\delta_{12}}+\sqrt{\delta_{24}}
    =
    \sqrt{-\delta_{13}}+\sqrt{-\delta_{34}}.
\]
Thus, in this sign pattern, a positive solution exists if and only if
\[
    \sqrt{\delta_{12}}+\sqrt{\delta_{24}}
    =
    \sqrt{-\delta_{13}}+\sqrt{-\delta_{34}},
    \qquad
    \sqrt{-\delta_{13}}>\sqrt{\delta_{12}}.
\]
In that case, the positive solution is unique up to scaling and is given by
\[
    (x_1,x_2,x_3)
    =
    \lambda
    \left(
        \sqrt{\delta_{12}},\,
        \sqrt{-\delta_{13}}-\sqrt{\delta_{12}},\,
        \sqrt{-\delta_{34}}
    \right),
    \qquad \lambda>0.
\]

Similarly, assume that
\[
    \delta_{12}<0,\qquad
    \delta_{24}<0,\qquad
    \delta_{13}>0,\qquad
    \delta_{34}>0.
\]
Then the three equations \eqref{tilde.G1}-\eqref{tilde.G3} imply
\[
    x_1+x_2
    =
    \frac{\sqrt{\delta_{13}}}{\sqrt{-\delta_{12}}}\,x_1,\quad
    x_2+x_3
    =
    \frac{\sqrt{-\delta_{24}}}{\sqrt{-\delta_{12}}}\,x_1,\quad
    x_3
    =
    \frac{\sqrt{\delta_{34}}}{\sqrt{-\delta_{12}}}\,x_1.
\]
Therefore a positive solution exists if and only if
\[
    \sqrt{-\delta_{12}}+\sqrt{-\delta_{24}}
    =
    \sqrt{\delta_{13}}+\sqrt{\delta_{34}},
    \qquad
    \sqrt{\delta_{13}}>\sqrt{-\delta_{12}}.
\]
In this case, the positive solution is unique up to scaling and is given by
\[
    (x_1,x_2,x_3)
    =
    \lambda
    \left(
        \sqrt{-\delta_{12}},\,
        \sqrt{\delta_{13}}-\sqrt{-\delta_{12}},\,
        \sqrt{\delta_{34}}
    \right),
    \qquad \lambda>0.
\]
We summarize this case of collinear four-body equilibrium as follows.
\begin{proposition}[A collinear four-body reduction]
Consider the charged four-body problem in the collinear order $a_1<a_2<a_3<a_4$,
and set $x_1=a_2-a_1,\;x_2=a_3-a_2,\;x_3=a_4-a_3$.
Assume that
\[
\delta_{14}=\delta_{23}=0.
\]
Then a non-collision collinear equilibrium exists if and only if one of the
following two alternatives holds.

\begin{itemize}
\item[(i)] $\delta_{12}>0,\;\delta_{24}>0,\;\delta_{13}<0,\;\delta_{34}<0$
and
\[
    \sqrt{\delta_{12}}+\sqrt{\delta_{24}}
    =
    \sqrt{-\delta_{13}}+\sqrt{-\delta_{34}},
    \qquad
    \sqrt{-\delta_{13}}>\sqrt{\delta_{12}}.
\]
In this case the positive solution is unique up to scaling and is given by
\[
    (x_1,x_2,x_3)
    =
    \lambda
    \left(
        \sqrt{\delta_{12}},
        \sqrt{-\delta_{13}}-\sqrt{\delta_{12}},
        \sqrt{-\delta_{34}}
    \right),
    \qquad
    \lambda>0.
\]

\item[(ii)] $\delta_{12}<0,\;\delta_{24}<0,\;\delta_{13}>0,\;\delta_{34}>0$
and
\[
    \sqrt{-\delta_{12}}+\sqrt{-\delta_{24}}
    =
    \sqrt{\delta_{13}}+\sqrt{\delta_{34}},
    \qquad
    \sqrt{\delta_{13}}>\sqrt{-\delta_{12}}.
\]
In this case the positive solution is unique up to scaling and is given by
\[
    (x_1,x_2,x_3)
    =
    \lambda
    \left(
        \sqrt{-\delta_{12}},
        \sqrt{\delta_{13}}-\sqrt{-\delta_{12}},
        \sqrt{\delta_{34}}
    \right),
    \qquad
    \lambda>0.
\]
\end{itemize}
\end{proposition}

\subsection{The inverse problem of collinear equilibrium}

In \cite{Mou}, Moulton also
considered the inverse problem: given a collinear configuration, find the mass vectors, if any, for which
it is a central configuration. This becomes a system of linear algebraic equations and 
the possible mass vectors are determined by the Pfaffians of the associated matrices.

Inspired by Moulton's inverse problem for collinear central configurations, we
now consider the corresponding inverse problem for collinear equilibria in the
charged \(N\)-body problem. Namely, for a prescribed collinear configuration, we
ask whether one can choose masses and charges so that the given configuration
becomes a static equilibrium.

Recall $\lambda_i=\frac{e_i}{m_i},\; 1\leq i\leq N$,
and let
\[
    M=\operatorname{diag}(m_1,\ldots,m_N),
    \qquad
    \Lambda=\operatorname{diag}(\lambda_1,\ldots,\lambda_N).
\]
For a fixed collinear configuration $a=(a_1,\ldots,a_N)$
with distinct coordinates $a_i\neq a_j$ for $i\ne j$,
define the skew-symmetric matrix \(A=(A_{ij})\) by
\[
    A_{ij}=
    \begin{cases}
    \dfrac{1}{|a_i-a_j|^2}, & i<j,\\[1.2ex]
    -\dfrac{1}{|a_i-a_j|^2}, & i>j,\\[1.2ex]
    0, & i=j,
    \end{cases}
\]
that is, 
\begin{equation}\label{A}
A = \begin{pmatrix}
0 & \frac{1}{|a_1-a_2|^2} & \frac{1}{|a_1-a_3|^2} & \cdots & \frac{1}{|a_1-a_N|^2} \\
-\frac{1}{|a_1-a_2|^2} & 0 & \frac{1}{|a_2-a_3|^2} & \cdots & \frac{1}{|a_2-a_N|^2} \\
-\frac{1}{|a_1-a_3|^2} & -\frac{1}{|a_2-a_3|^2} & 0 & \ddots & \vdots \\
\vdots & \vdots & \vdots & \ddots & \frac{1}{|a_{N-1}-a_N|^2} \\
-\frac{1}{|a_1-a_N|^2} & -\frac{1}{|a_2-a_N|^2} & \cdots & -\frac{1}{|a_{N-1}-a_N|^2} & 0
\end{pmatrix}.
\end{equation}
Then the equilibrium equation \eqref{eq.of.e} can be written as
\begin{equation*}
    M(\Lambda A\Lambda-A)
    \begin{pmatrix}
    m_1\\
    \vdots\\
    m_N
    \end{pmatrix}
    =0.
\end{equation*}
Since \(M\) is invertible, this is equivalent to
\begin{equation}\label{eq.of.collinear.equilibrium}
    (\Lambda A\Lambda-A)m=0,
\end{equation}
where $m=(m_1,\ldots,m_N)^T$.
For later use, define
\[
F(\lambda;m)=(\Lambda A\Lambda-A)m,
\]
and
\begin{equation}\label{F}
    F_i(\lambda;m)
    =
    \sum_{j\neq i}A_{ij}m_j(\lambda_i\lambda_j-1),
\end{equation}
for $1\leq i\leq N$.
Then the equation \eqref{eq.of.collinear.equilibrium} is equivalent to
\[
    F(\lambda;m)=0,
\]
where $\lambda=(\lambda_1,\ldots,\lambda_N)$.
Let
\begin{equation}\label{J}
    J(a,m)
    =
    \left.
    \frac{\partial F}{\partial \lambda}
    \right|_{\lambda=(1,\ldots,1)}
\end{equation}
be the Jacobian matrix of \(F\) at the trivial solution $\lambda_1=\cdots=\lambda_N=1$.
A direct calculation gives
\begin{equation}\label{J_il}
    J_{i\ell}(a,m)
    =
    \begin{cases}
    \displaystyle
    \sum_{j\neq i}A_{ij}m_j, & \ell=i,\\[2ex]
    A_{i\ell}m_\ell, & \ell\neq i.
    \end{cases}
\end{equation}

The following criterion is
then a direct rank condition for the inverse problem.

\begin{lemma}
\label{lem:inverse-rank-condition}
Fix positive masses $m=(m_1,\ldots,m_N)\in(\mathbb R^+)^N$
and a collinear configuration
\[
    a=(a_1,\ldots,a_N),
\]
with distinct coordinates $a_i\neq a_j$ for $i\ne j$.
If
\begin{equation}\label{rank.condition}
    \operatorname{rank}J(a,m)=N-1,
\end{equation}
then there exist real charges $e_1,\ldots,e_N\in\mathbb R$
such that the prescribed collinear configuration \(a\) is an equilibrium of the
charged \(N\)-body problem.
\end{lemma}

\begin{proof}
The equation \eqref{eq.of.collinear.equilibrium}
is equivalent to
\[
    F(\lambda;m)=0.
\]
The point $\lambda_0=(1,\ldots,1)$
is always a solution. It corresponds to the trivial case $\Lambda A\Lambda=A$,
or equivalently, $\delta_{ij}=m_im_j-e_ie_j=0$ for all $i\ne j$.

We show that, under the rank assumption \eqref{rank.condition}, this trivial solution is accompanied
by nearby non-trivial solutions. Since \(A\) is skew-symmetric, one has
\[
    m^T F(\lambda;m)=0
\]
for all \(\lambda\). Hence the \(N\) equations \(F_i=0\) have at most
\(N-1\) independent equations near \(\lambda_0\).

By rank assumption \eqref{rank.condition},
after deleting one suitable component of \(F\), say \(F_k\), the
remaining \(N-1\) equations have full rank at \(\lambda_0\). By the implicit
function theorem, their common zero set near \(\lambda_0\) contains a
one-dimensional curve of real solutions passing through \(\lambda_0\).

Moreover, every point on this curve also satisfies the deleted equation
\(F_k=0\). Indeed, since $\sum_{i=1}^N m_iF_i(\lambda;m)=0$,
if \(F_i(\lambda;m)=0\) for all \(i\neq k\), then, since \(m_k>0\), it follows
that
\[
    F_k(\lambda;m)=0.
\]
Thus the original equation $F(\lambda;m)=0$
has real solutions arbitrarily close to \(\lambda_0\) and different from
\(\lambda_0\).

For such a nearby solution, we have
\[
    \Lambda A\Lambda\neq A.
\]
Indeed, if \(\Lambda A\Lambda=A\), then, since \(A_{ij}\neq0\),
we have $\lambda_i\lambda_j=1$ for all $i\neq j$.
For \(N\geq3\), this implies $\lambda_1=\cdots=\lambda_N=1$
in a sufficiently small neighbourhood of \(\lambda_0\). Hence every nearby
solution different from \(\lambda_0\) is non-trivial.

Finally, define $e_i=m_i\lambda_i,\; 1\leq i\leq N$.
Then $\delta_{ij}=m_im_j-e_ie_j
    =
    m_im_j(1-\lambda_i\lambda_j)$,
and the equation \eqref{eq.of.collinear.equilibrium}
is precisely the equilibrium equation for the charged \(N\)-body problem.
Therefore, the prescribed collinear configuration is realized as a non-trivial
equilibrium.
\end{proof}

\begin{remark}
The rank condition in Lemma~\ref{lem:inverse-rank-condition} is only a
sufficient condition. If it fails, non-trivial solutions of \eqref{eq.of.collinear.equilibrium}
may still exist. 
The lemma only provides a simple condition which guarantees
their existence near the trivial solution $\lambda_1=\cdots=\lambda_N=1$.
\end{remark}

The preceding lemma shows that, for fixed positive masses, the prescribed
collinear configuration can be realized once the matrix \(J(a,m)\) has the
maximal possible rank \(N-1\). Thus it remains to see whether this rank condition
can be fulfilled by choosing the masses appropriately.

The following theorem gives an affirmative answer. It shows that every
collinear configuration of distinct points can be realized as a non-trivial
equilibrium of a suitable charged \(N\)-body problem. Moreover, for the fixed
configuration, the exceptional choices of masses are contained in a proper
algebraic subset of \((\mathbb R^+)^N\).

\begin{theorem}[Realization of prescribed collinear configurations]
\label{thm:realization-collinear}
For $N\ge3$, let
\[
    a_1<a_2<\cdots<a_N,
\]
be any prescribed collinear configuration of \(N\) distinct points. Then one
can choose positive masses $(m_1,\ldots,m_N)\in(\mathbb{R}^+)^N$
and real charges $(e_1,\ldots,e_N)\in(\mathbb R)^N$
such that the prescribed configuration is a non-trivial equilibrium of the
charged \(N\)-body problem.
\end{theorem}
\begin{proof}
Let \(A=(A_{ij})\) be the skew-symmetric matrix associated with the prescribed
configuration, as defined by \eqref{A}. Recall that the equilibrium equation can be written as
\[
    F(\lambda;m)=0.
\]
where $F(\lambda;m)$ is given by \eqref{F}.
We now show that, for a generic choice of positive masses, this trivial
solution is contained in a non-trivial local branch of solutions. 
As a similar argument in the proof of Lemma \ref{lem:inverse-rank-condition},
one has
\[
    m^T F(\lambda;m)=0
\]
for all \(\lambda\). Hence
\[
    \operatorname{rank}J(a,m)\leq N-1,
\]
where $J(a,m)$ is given by \eqref{J}.
It remains to see that the rank is equal to \(N-1\) for a generic choice of
\(m\). 

Consider the \((N-1)\times (N-1)\) minor obtained by deleting the last
row and the last column:
\[
    D_a(m)
    =
    \det\left(J_{i\ell}\right)_{1\leq i,\ell\leq N-1},
\]
where $J_{il}$ is given by \eqref{J_il}.
This is a polynomial in \(m_1,\ldots,m_N\). As a polynomial in \(m_N\), its
highest-degree term comes from the product of the diagonal terms
\[
    A_{1N}m_N,\ A_{2N}m_N,\ \ldots,\ A_{N-1,N}m_N.
\]
Therefore the coefficient of \(m_N^{N-1}\) in \(D_a(m)\) is
\[
    \prod_{i=1}^{N-1} A_{iN}
    =
    \prod_{i=1}^{N-1}\frac{1}{|a_i-a_N|^2},
\]
which is nonzero because the points \(a_1,\ldots,a_N\) are distinct. Hence
\(D_a(m)\) is not the zero polynomial.

Let
\[
    \Sigma_a
    =
    \{m\in(\mathbb R^+)^N\mid D_a(m)=0\}.
\]
Then \(\Sigma_a\) is contained in a proper algebraic subset of
\((\mathbb R^+)^N\). For every $m\in(\mathbb R^+)^N\setminus\Sigma_a$,
we have
\[
    \operatorname{rank}J(a,m)=N-1.
\]

By Lemma~\ref{lem:inverse-rank-condition}, the rank condition implies that
the equation \eqref{eq.of.collinear.equilibrium}
has non-trivial real solutions arbitrarily close to $\lambda_1=\cdots=\lambda_N=1$.
Choose one such solution
\[
    \lambda=(\lambda_1,\ldots,\lambda_N).
\]
Then, by setting $e_i=m_i\lambda_i,\; 1\leq i\leq N$,
we obtain real charges for which the prescribed collinear configuration
$(a_1,a_2,\cdots,a_N)$
is an equilibrium of the charged \(N\)-body problem.

Finally, the solution is non-trivial because $\Lambda A\Lambda\neq A$.
This completes the proof.
\end{proof}

\subsection{Special Case: \( N = 4 \)}

We now apply the rank criterion to the case \(N=4\). Let
\[
    a_1<a_2<a_3<a_4
\]
be a prescribed collinear configuration. 
The skew-symmetric matrix \(A\) is given by
\[
    A_{ij}=\frac{1}{(a_j-a_i)^2},
    \qquad
    A_{ji}=-\frac{1}{(a_j-a_i)^2}.
\]
for $1\le i<j\le4$.

We consider the Jacobian matrix of \(F\) with respect to \(\lambda\) at
\(\lambda_0=(1,1,1,1)\):
\[
    J(a,m)
    =
    \left.
    \frac{\partial F}{\partial\lambda}
    \right|_{\lambda=\lambda_0}.
\]
A direct computation gives
{\footnotesize
\begin{eqnarray*}
&&J(a,m)\\
&&=
\begin{pmatrix}
 A_{12}m_2+A_{13}m_3+A_{14}m_4
 & A_{12}m_2
 & A_{13}m_3
 & A_{14}m_4
\\[0.6ex]
 -A_{12}m_1
 & -A_{12}m_1+A_{23}m_3+A_{24}m_4
 & A_{23}m_3
 & A_{24}m_4
\\[0.6ex]
 -A_{13}m_1
 & -A_{23}m_2
 & -A_{13}m_1-A_{23}m_2+A_{34}m_4
 & A_{34}m_4
\\[0.6ex]
 -A_{14}m_1
 & -A_{24}m_2
 & -A_{34}m_3
 & -A_{14}m_1-A_{24}m_2-A_{34}m_3
\end{pmatrix}.
\end{eqnarray*}
}

The following proposition is the \(N=4\) version of the local rank criterion.

\begin{proposition}
\label{prop:N4-rank-criterion}
Let $a_1<a_2<a_3<a_4$
be a prescribed collinear configuration. If $m=(m_1,m_2,m_3,m_4)\in(\mathbb R^+)^4$,
define
\begin{eqnarray}\label{D}
\mathcal D(a,m)&=&
-\frac{m_1^2}
{(a_1-a_2)^2(a_1-a_3)^2(a_1-a_4)^2}
+\frac{m_2^2}
{(a_2-a_1)^2(a_2-a_3)^2(a_2-a_4)^2}
\nonumber\\
&&-\frac{m_3^2}
{(a_3-a_1)^2(a_3-a_2)^2(a_3-a_4)^2}
+\frac{m_4^2}
{(a_4-a_1)^2(a_4-a_2)^2(a_4-a_3)^2}.
\end{eqnarray}

If $\mathcal D(a,m)\neq 0$,
then there exist real charges $(e_1,e_2,e_3,e_4)\in\mathbb R^4$
such that the prescribed collinear configuration $(a_1,a_2,a_3,a_4)$
is a non-trivial equilibrium of the charged four-body problem.

In particular, for every prescribed collinear configuration of four distinct
points, one can choose positive masses and real charges so that the
configuration is realized as a non-trivial equilibrium. More precisely, for
the fixed configuration \(a\), all positive masses outside the proper algebraic
set
\[
    \{m\in(\mathbb R^+)^4\mid \mathcal D(a,m)=0\}
\]
satisfy the above sufficient condition.
\end{proposition}

\begin{proof}
By Lemma~\ref{lem:inverse-rank-condition} and Theorem \ref{thm:realization-collinear}, 
it is enough to prove that
\[
    \operatorname{rank}J(a,m)=3.
\]
Since \(M=\operatorname{diag}(m_1,m_2,m_3,m_4)\) is invertible, this is
equivalent to
\[
    \operatorname{rank}MJ(a,m)=3.
\]

For \(1\leq i<j\leq4\), set
\begin{equation}\label{c_ij}
    c_{ij}:=m_im_jA_{ij}=\frac{m_im_j}{(a_j-a_i)^2}.
\end{equation}
Then \(c_{ij}>0\), and a direct computation gives
\[
MJ(a,m)=
\begin{pmatrix}
c_{12}+c_{13}+c_{14} & c_{12} & c_{13} & c_{14}\\
-c_{12} & -c_{12}+c_{23}+c_{24} & c_{23} & c_{24}\\
-c_{13} & -c_{23} & -c_{13}-c_{23}+c_{34} & c_{34}\\
-c_{14} & -c_{24} & -c_{34} & -c_{14}-c_{24}-c_{34}
\end{pmatrix}.
\]
The off-diagonal part of this matrix is skew-symmetric, and its diagonal
entries sum to zero. Moreover,
\[
    \mathbf 1^T MJ(a,m)=0,
\]
because \(m^T J(a,m)=0\). Hence \(0\) is always an eigenvalue of \(MJ(a,m)\).

Let
\[
    \chi(\mu)=\det(\mu I-MJ(a,m))
\]
be the characteristic polynomial of \(MJ(a,m)\). A direct computation gives
\[
    \chi(\mu)=\mu\bigl(\mu^3+P\mu+Q\bigr),
\]
where the explicit expression of \(P\) is not needed here, and
\begin{equation}\label{Q}
    Q
    =
    2\left(
    -c_{12}c_{13}c_{14}
    +c_{12}c_{23}c_{24}
    -c_{13}c_{23}c_{34}
    +c_{14}c_{24}c_{34}
    \right).
\end{equation}
Substituting \eqref{c_ij} into \eqref{Q}, we obtain
\[
    Q
    =
    2m_1m_2m_3m_4\,\mathcal D(a,m),
\]
where $\mathcal D(a,m)$ is given by \eqref{D}.
Therefore the assumption $\mathcal D(a,m)\neq0$
implies $Q\neq0$.
Hence \(0\) is a simple root of \(\chi(\mu)\). Consequently,
\[
    \dim\ker(MJ(a,m))=1,
\]
and therefore
\[
    \operatorname{rank}MJ(a,m)=3.
\]
Since \(M\) is invertible, this gives $\operatorname{rank}J(a,m)=3$.
\end{proof}

\section{Equilibrium of the Regular $n$-gon Configuration}

We now turn from collinear configurations to a highly symmetric class of
non-collinear configurations. More precisely, we consider the case where the
particles are placed at the vertices of a regular \(n\)-gon. This setting is
natural because the geometry of the configuration is completely determined by
the cyclic symmetry, while the masses and charges may still vary from vertex to
vertex. Thus the equilibrium equations reduce to algebraic conditions on the
effective interaction coefficients $\delta_{ij}$ for $1\le i\ne j\le n$.

We divide the discussion into two cases. We first consider the regular
\(n\)-gon formed only by the \(n\) particles on the vertices. Then we add one
more particle at the geometric center and study how the equilibrium equations
are modified.

\subsection{Regular polygon without a central particle}

Before writing down the full system of balance equations, we record a simple
but useful observation. If one particle has zero charge-to-mass ratio, then the
interactions acting on it are purely attractive. Since the other vertices of a
convex polygon all lie in the same open half-plane determined by the tangent
line at that vertex, such attractive forces cannot cancel each other. Therefore
each charge-to-mass ratio must be nonzero.
\begin{lemma}\label{Lm:rank.condition.for.n-gon}
If the equilibrium configuration forms a convex $n$-gon, then $\lambda_i= \frac{e_i}{m_i} \neq 0$ for all $i$.
\end{lemma}


Now we place a regular $n$-gon on the complex plane, with its geometric center at the origin and one vertex at $1$. The other vertices lie at $\omega, \omega^2, \dots, \omega^{n-1}$, where $\omega = e^{2\pi i/n}$.

For the particle at position $1$, its force balance equations are:
\begin{align}
      &(m_1m_2 - e_1e_2)a_2 + (m_1m_3 - e_1e_3)a_3 + \dots + (m_1m_n - e_1e_n)a_n = 0, \\
&(m_1m_2 - e_1e_2)b_2 + (m_1m_3 - e_1e_3)b_3 + \dots + (m_1m_n - e_1e_n)b_n = 0,
\end{align}

where
\begin{equation}
    a_i = \operatorname{Re}\left(\frac{1-\omega^{i-1}}{|1-\omega^{i-1}|^3}\right),\quad
b_i = \operatorname{Im}\left(\frac{1-\omega^{i-1}}{|1-\omega^{i-1}|^3}\right),\quad
i=2,3,\dots,n.\label{a_i.b_i}
\end{equation}

These represent the force components along the real and imaginary axes, respectively.

By rotating particle $2$ to particle $1$, we obtain the force balance equations for particle $2$:
\begin{align}
    (m_2m_3 - e_2e_3)a_2 + (m_2m_4 - e_2e_4)a_3 + \dots + (m_2m_n - e_2e_n)a_{n-1} + (m_2m_1 - e_2e_1)a_n = 0, \\
(m_2m_3 - e_2e_3)b_2 + (m_2m_4 - e_2e_4)b_3 + \dots + (m_2m_n - e_2e_n)b_{n-1} + (m_2m_1 - e_2e_1)b_n = 0.
\end{align}

The same procedure yields the balance equations for all other particles. We now consolidate all equations into a matrix form.

\begin{figure}
    \centering
    \includegraphics[width=0.5\linewidth]{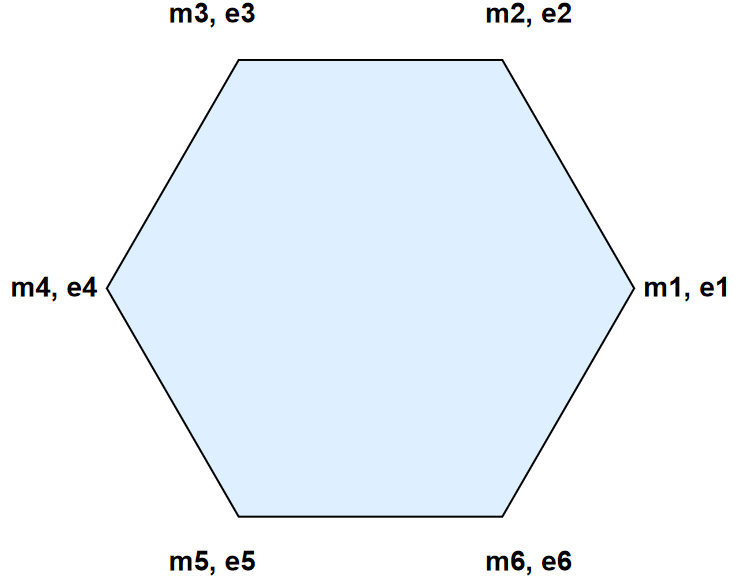}
    \caption{n=6}
    \label{fig:placeholder}
\end{figure}

The full system of equations can be written as:
\begin{equation}
    \begin{pmatrix}
m_1m_2 - e_1e_2 & m_1m_3 - e_1e_3 & \dots & m_1m_n - e_1e_n \\
m_2m_3 - e_2e_3 & m_2m_4 - e_2e_4 & \dots & m_2m_1 - e_2e_1 \\
\vdots & \vdots & \ddots & \vdots \\
m_nm_1 - e_ne_1 & m_nm_2 - e_ne_2 & \dots & m_nm_{n-1} - e_ne_{n-1}
\end{pmatrix}
\begin{pmatrix}
a_2 & b_2 \\
a_3 & b_3 \\
\vdots & \vdots \\
a_n & b_n
\end{pmatrix}
= 0.
\end{equation}

Setting $a_1 = b_1 = 0$, this becomes:
\begin{equation}
    \begin{pmatrix}
m_1^2 - e_1^2 & m_1m_2 - e_1e_2 & \dots & m_1m_n - e_1e_n \\
m_2^2 - e_2^2 & m_2m_3 - e_2e_3 & \dots & m_2m_1 - e_2e_1 \\
\vdots & \vdots & \ddots & \vdots \\
m_n^2 - e_n^2 & m_nm_1 - e_ne_1 & \dots & m_nm_{n-1} - e_ne_{n-1}
\end{pmatrix}
\begin{pmatrix}
a_1 & b_1 \\
a_2 & b_2 \\
\vdots & \vdots \\
a_n & b_n
\end{pmatrix}
= 0.
\end{equation}

Let $\lambda_i = \dfrac{e_i}{m_i}$. We can rewrite the above as:
\begin{equation}
    \begin{pmatrix}
m_1 & & 0 \\
& m_2 & \\
0 & & \vdots \\
& & m_n
\end{pmatrix}
\begin{pmatrix}
m_1 - \lambda_1e_1 & m_2 - \lambda_1e_2 & \dots & m_n - \lambda_1e_n \\
m_2 - \lambda_2e_2 & m_3 - \lambda_2e_3 & \dots & m_1 - \lambda_2e_1 \\
\vdots & \vdots & \ddots & \vdots \\
m_n - \lambda_ne_n & m_1 - \lambda_ne_1 & \dots & m_{n-1} - \lambda_ne_{n-1}
\end{pmatrix}
\begin{pmatrix}
a_1 & b_1 \\
a_2 & b_2 \\
\vdots & \vdots \\
a_n & b_n
\end{pmatrix}
= 0.
\end{equation}

Since $\operatorname{diag}(m_1,m_2,\dots,m_n)$ is non-singular, we obtain the matrix equation for $\mathbf{m} = (m_1,m_2,\dots,m_n)^T$:
\begin{equation}
    \left[
\begin{pmatrix}
m_1 & m_2 & \dots & m_n \\
m_2 & m_3 & \dots & m_1 \\
\vdots & \vdots & \ddots & \vdots \\
m_n & m_1 & \dots & m_{n-1}
\end{pmatrix}
- \Lambda
\begin{pmatrix}
m_1 & m_2 & \dots & m_n \\
m_2 & m_3 & \dots & m_1 \\
\vdots & \vdots & \ddots & \vdots \\
m_n & m_1 & \dots & m_{n-1}
\end{pmatrix}
\Lambda
\right]
\begin{pmatrix}
a_1 & b_1 \\
a_2 & b_2 \\
\vdots & \vdots \\
a_n & b_n
\end{pmatrix}
= 0,
\end{equation}

where $\Lambda = \operatorname{diag}(\lambda_1,\lambda_2,\dots,\lambda_n)$. This splits into two separate matrix equations:
\begin{equation}
    \left[
\begin{pmatrix}
a_1 & a_2 & \dots & a_n \\
a_n & a_1 & \dots & a_{n-1} \\
\vdots & \vdots & \ddots & \vdots \\
a_2 & a_3 & \dots & a_1
\end{pmatrix}
- \Lambda
\begin{pmatrix}
a_1 & a_2 & \dots & a_n \\
a_n & a_1 & \dots & a_{n-1} \\
\vdots & \vdots & \ddots & \vdots \\
a_2 & a_3 & \dots & a_1
\end{pmatrix}
\Lambda
\right]
\mathbf{m}
= 0,
\end{equation}

\begin{equation}
    \left[
\begin{pmatrix}
b_1 & b_2 & \dots & b_n \\
b_n & b_1 & \dots & b_{n-1} \\
\vdots & \vdots & \ddots & \vdots \\
b_2 & b_3 & \dots & b_1
\end{pmatrix}
- \Lambda
\begin{pmatrix}
b_1 & b_2 & \dots & b_n \\
b_n & b_1 & \dots & b_{n-1} \\
\vdots & \vdots & \ddots & \vdots \\
b_2 & b_3 & \dots & b_1
\end{pmatrix}
\Lambda
\right]
\mathbf{m}
= 0.
\end{equation}

Define the symmetric circulant matrix $A$ and skew-symmetric circulant matrix $B$:
\begin{equation}
    A =
\begin{pmatrix}
a_1 & a_2 & \dots & a_n \\
a_n & a_1 & \dots & a_{n-1} \\
\vdots & \vdots & \ddots & \vdots \\
a_2 & a_3 & \dots & a_1
\end{pmatrix},\quad
B =
\begin{pmatrix}
b_1 & b_2 & \dots & b_n \\
b_n & b_1 & \dots & b_{n-1} \\
\vdots & \vdots & \ddots & \vdots \\
b_2 & b_3 & \dots & b_1
\end{pmatrix}.\label{eq:1}
\end{equation}

Then the equations become:
\begin{equation}
    (A - \Lambda A \Lambda)\mathbf{m} = 0,\quad
(B - \Lambda B \Lambda)\mathbf{m} = 0.
\end{equation}

\begin{proposition}\label{prop:equilibrium.eq.of.n-gon}

    The equilibrium conditions for the regular n-gon configuration are
    \[
(A - \Lambda A \Lambda)\mathbf{m} = 0,\quad
(B - \Lambda B \Lambda)\mathbf{m} = 0,
\]
where $\Lambda = \operatorname{diag}(\lambda_1,\lambda_2,\dots,\lambda_n)$ and $A,B$ are given by \eqref{eq:1}.
\end{proposition}

\begin{lemma}\label{lem:null.dm.of.B}
Let \(B\) be the skew-symmetric circulant matrix defined in \eqref{eq:1}, namely
the matrix associated with the tangential components of the force equations
for the regular \(n\)-gon. Then \(0\) is an eigenvalue of \(B\), and the
corresponding eigenspace is as follows.
\begin{itemize}
    \item If \(n\) is odd, then
    \[
        \ker B=\operatorname{span}\{\mathbf 1\},
    \]
    where $\mathbf 1=(1,1,\ldots,1)^T$.
    In particular, \(0\) is a simple eigenvalue of \(B\).

    \item If \(n\) is even, then
    \[
        \ker B=\operatorname{span}\{\mathbf 1,\mathbf v\},
    \]
    where
    \[
        \mathbf 1=(1,1,\ldots,1)^T,
        \qquad
        \mathbf v=(1,-1,1,-1,\ldots,1,-1)^T.
    \]
    In particular, the multiplicity of the zero eigenvalue of \(B\) is \(2\).
\end{itemize}
\end{lemma}
\begin{proof}
    The proof of this lemma will be deferred to the Appendix.
\end{proof}

Next, we consider the case of equal masses, where the conclusions will be further simplified.

\begin{theorem}\label{thm:nonexistence.of.n-gon}
    Suppose $n$ particles with equal masses form a regular $n$-gon equilibrium, then their electric charges must be the same, and the absolute value of each charge equals the mass of the particle, i.e.,
    \[
|e_1| = |e_2| = \dots = |e_n|=
m_1 = m_2 = \dots = m_n.
\]
In other words, under the non-trivial assumption, $n$ equal masses cannot form a regular $n$-gon equilibrium.
\end{theorem}
\begin{proof} Since the masses are equal, we may write
\[
    m=(m_0,\ldots,m_0)^T=m_0\mathbf 1,
\]
for some $m_0>0$.
The equilibrium equations are homogeneous in the mass vector \(m\), and hence
we may divide by \(m_0\) and work with
\[
    m=\mathbf 1=(1,\ldots,1)^T.
\]
By Lemma~\ref{lem:null.dm.of.B}, the kernel of the skew-symmetric circulant matrix \(B\) depends
on the parity of \(n\). We therefore distinguish two cases.

{\bf Case 1: $n$ is odd.}
When $n$ is odd, $B$ has exactly one zero eigenvalue, with eigenvector ${\mathbf{1}} = (1,1,\dots,1)^T$.
From $(B - \Lambda B \Lambda){\mathbf{m}} = 0$, substituting ${\mathbf{m}} = {\mathbf{1}}$ gives:
\begin{equation}
    B{\mathbf{1}} - \Lambda B \Lambda {\mathbf{1}} = 0.
\end{equation}

Since $B{\mathbf{1}} = 0$, we have $\Lambda B \Lambda {\mathbf{1}} = 0$. As $\Lambda$ is non-singular, $B \Lambda{\mathbf{1}} = 0$, so $\Lambda {\mathbf{1}} = (\lambda_1,\lambda_2,\dots,\lambda_n)^T$ is an eigenvector of $B$ corresponding to the zero eigenvalue. Thus, $\lambda_1 = \lambda_2 = \dots = \lambda_n = \lambda$.

Substituting into $(A - \Lambda A \Lambda)\mathbf{m} = 0$:
\begin{equation}
    (1 - \lambda^2)A\mathbf{m} = 0.
\end{equation}

Since $A$ has non-negative entries and $\mathbf{m} = (1,1,\dots,1)^T > 0$, we must have $1 - \lambda^2 = 0$, so $\lambda = \pm 1$.

{\bf Case 2: $n$ is even.}
When $n$ is even, $B$ has two zero eigenvalues, with eigenvectors ${\mathbf{1}} = (1,1,\dots,1)^T$ and ${v}=(1,-1,1,\dots,-1)^T$.
Similarly, we find $B\Lambda\mathbf{1} = 0$, so $\Lambda\mathbf{1}$ is a zero eigenvector of $B$. Let
\begin{equation}
\Lambda\mathbf{1}  = k_1 \mathbf{1} + k_2 {v} =
\begin{pmatrix}
k_1 + k_2 \\
k_1 - k_2 \\
k_1 + k_2 \\
\vdots \\
k_1 - k_2
\end{pmatrix} =
\begin{pmatrix}
\lambda_1 \\
\lambda_2 \\
\lambda_3 \\
\vdots \\
\lambda_n
\end{pmatrix},
\quad \text{assuming } k_2 \neq 0.
\end{equation}

From \((A - \Lambda A \Lambda)\mathbf{1} = 0\):
\begin{align*}
0 &= \left( \sum_{i=1}^n a_i \right) \mathbf{1} - \Lambda A \left( k_1 \mathbf{1} + k_2 {v} \right) \\
&= \left( \sum_{i=1}^n a_i \right) \mathbf{1} - \left( k_1 \sum_{i=1}^n a_i \right) \Lambda \mathbf{1} - k_2 \Lambda A {v} \\
&= \left( \sum_{i=1}^n a_i \right)\mathbf{1}  - \left( k_1 \sum_{i=1}^n a_i \right) \left( k_1\mathbf{1} + k_2 {v} \right) - k_2 \Lambda A {v} \\
&= \left( \sum_{i=1}^n a_i \right) \left( 1 - k_1^2 \right) \mathbf{1}- k_2 \left[ k_1 \left( \sum_{i=1}^n a_i \right) {v} + \left( \sum_{i=1}^n (-1)^{i-1} a_i \right) \Lambda {v} \right].
\end{align*}

Define
 \begin{equation}\label{a.ahat}
     a = \sum_{i=1}^n a_i, \quad \hat{a} = \sum_{i=1}^n (-1)^{i-1} a_i.
 \end{equation}

Then
\begin{equation}
    a(1 - k_1^2) \mathbf{1} - k_2 \left[ k_1 a v + \hat{a} \Lambda {v} \right] = 0,
\end{equation}

where
\[
k_1 a {v} + \hat{a} \Lambda {v} =
\begin{pmatrix}
k_1 a \\
-k_1 a \\
\vdots \\
-k_1 a
\end{pmatrix} +
\begin{pmatrix}
\lambda_1 \hat{a} \\
-\lambda_2 \hat{a} \\
\vdots \\
-\lambda_n \hat{a}
\end{pmatrix}.
\]

Considering the first two rows:
\[
\begin{cases}
k_1 a + \lambda_1 \hat{a} = \dfrac{a(1 - k_1^2)}{k_2} \\
-k_1 a - \lambda_2 \hat{a} = \dfrac{a(1 - k_1^2)}{k_2}
\end{cases}
\implies
\begin{cases}
a(1 - k_1 k_2 - k_1^2) = k_2 (k_1 + k_2) \hat{a} \\
a(1 + k_1 k_2 - k_1^2) = -k_2 (k_1 - k_2) \hat{a}.
\end{cases}
\]

Since \(a_i > 0\), we have \(a > 0\).

\begin{enumerate}
\item If \(\hat{a} = 0\), then \(k_1 k_2 = 0\). Also, \(1 - k_1 k_2 - k_1^2 = 0 \implies 1 = 0\), a contradiction.
\item If \(\hat{a} \neq 0\), then
\[
\frac{1 - k_1 k_2 - k_1^2}{1 + k_1 k_2 - k_1^2} = \frac{k_1 + k_2}{k_2 - k_1}
\implies k_2^2 = k_1^2 - 1.
\]
\end{enumerate}

Substituting back, we get
\[
\begin{cases}
k_1 a + (k_1 + k_2) \hat{a} = -a k_2 \\
-k_1 a - (k_1 - k_2) \hat{a} = -a k_2
\end{cases}
\implies 2 k_1 a + 2 k_1 \hat{a} = 0
\implies k_1 (a + \hat{a}) = 0.
\]

But
\[
a + \hat{a} = \sum_{i=1}^n a_i + \sum_{i=1}^n (-1)^{i-1} a_i > 0,
\]
so \(k_1 = 0\), which implies \(k_2^2 = 0 - 1 = -1\), a contradiction.

In conclusion, \(k_2 = 0\), so \(\Lambda\mathbf{1} = k_1 \mathbf{1}\).
By the same argument as Case 1, we obtain \(\lambda = \pm 1\).
 \end{proof}

Before stating the four-body case, we point out a difference between the
general regular \(N\)-gon problem and the special case \(N=4\). Theorem~\ref{thm:nonexistence.of.n-gon} was proved under the equal-mass assumption. 
However, when \(N=4\), the
configuration is a square and hence is a concyclic quadrilateral. Therefore
the conclusion is stronger in this case: no condition on the masses is needed.
Indeed, Corollary~\ref{cor:degenerate-auxiliary-triangle} already shows that a non-trivial four-body equilibrium
cannot be concyclic. Thus the following result is an immediate consequence of
the four-body geometric obstruction obtained in Section~\ref{sec:4}.

\begin{theorem}\label{thm:square-equilibrium}
Suppose that four particles form a square equilibrium. Then
\[
    \delta_{ij}=0,\qquad 1\leq i<j\leq4.
\]
Equivalently, under the non-trivial assumption, no choice of four positive
masses and real charges can produce a square equilibrium.
\end{theorem}

\subsection{Regular polygon with one additional central particle}

If we consider a system of $(n+1)$ particles, where a particle with mass $m_0$ and charge $e_0$ is located at the center of a regular $n$-gon, then equations of the equilibrium become:
\begin{eqnarray}
    (m_1 m_2 - e_1 e_2) a_2 + (m_1 m_3 - e_1 e_3) a_3 + \dots + (m_1 m_n - e_1 e_n) a_n& + (m_0 m_1 - e_0 e_1) =& 0, 
\\
    (m_1 m_2 - e_1 e_2) b_2 + (m_1 m_3 - e_1 e_3) b_3 + \dots + (m_1 m_n - e_1 e_n) b_n &\qquad\qquad\qquad\qquad=& 0.
\end{eqnarray}
The full system of equations can be written as
\begin{equation}
\begin{pmatrix}
m_1 m_2 - e_1 e_2 & m_1 m_3 - e_1 e_3 & \dots & m_1 m_n - e_1 e_n \\
m_2 m_3 - e_2 e_3 & m_2 m_4 - e_2 e_4 & \dots & m_2 m_1 - e_2 e_1 \\
\vdots & \vdots & \ddots & \vdots \\
m_n m_1 - e_n e_1 & m_n m_2 - e_n e_2 & \dots & m_n m_{n-1} - e_n e_{n-1}
\end{pmatrix}
\begin{pmatrix}
a_2 & b_2 \\
a_3 & b_3 \\
\vdots & \vdots \\
a_n & b_n
\end{pmatrix}
=
\begin{pmatrix}
e_0 e_1 - m_0 m_1 & 0 \\
e_0 e_2 - m_0 m_2 & 0 \\
\vdots & \vdots \\
e_0 e_n - m_0 m_n & 0
\end{pmatrix}
    \end{equation}
which also can be transformed into:
\begin{equation}\label{eq.of.n+1-gon}
    (A - \Lambda A \Lambda) m = \left( e_0 \lambda_1 - m_0,\, e_0 \lambda_2 - m_0,\, \dots,\, e_0 \lambda_n - m_0 \right)^T, \quad
(B - \Lambda B \Lambda) m = 0,
\end{equation}
where \(A\) and \(B\), given by \eqref{eq:1}.

We now return to the centered regular \(n\)-gon with equal outer masses. 
Compared with the regular
\(n\)-gon without a central particle, the only new contribution comes from the
force exerted by the central particle, and this contribution is purely radial.
Hence the tangential equation is unchanged and still forces the charge-to-mass
ratios of the particles on the vertices to be highly constrained. After this
constraint is identified, the full equilibrium condition reduces to one scalar
equation involving the common charge-to-mass ratio and the mass and charge of
the central particle. The following proposition records this reduction.

\begin{theorem}[Centered regular \(n\)-gon with equal outer masses]
\label{thm:centered-regular-ngon}
Consider a regular \(n\)-gon with one additional particle at its center. 
Assume that the \(n\) particles on the vertices have the same mass
\[m_*=m_1=\cdots=m_n>0.\] 
Let $\lambda_i=\frac{e_i}{m_i}$ for $1\leq i\leq n$,
and let \(m_0>0\) and \(e_0\in\mathbb R\) be the mass and charge of the central
particle. Suppose also that
\[
    \lambda_i\neq0,\qquad 1\leq i\leq n.
\]
Then the centered regular \(n\)-gon is an equilibrium if and only if
\[
    \lambda_1=\lambda_2=\cdots=\lambda_n=:\lambda
\]
and
\begin{equation}\label{existence.condition}
    m^*(1-\lambda^2)\sum_{i=1}^n a_i=e_0\lambda-m_0.
\end{equation}
Here the constants \(a_i\) are those defined in \eqref{a_i.b_i}.
\end{theorem}
\begin{proof}
Substituting $\boldsymbol{m} = m^*\boldsymbol{1}$ into the equilibrium equations \eqref{eq.of.n+1-gon}, 
we obtain
\[
    (A-\Lambda A\Lambda)m^*\mathbf 1
    =
    (e_0\lambda_1-m_0,\ldots,e_0\lambda_n-m_0)^T,
\]
and
\[
    (B-\Lambda B\Lambda)m^*\mathbf 1=0.
\]
Since $B\mathbf 1=0$,
the second equation gives
\[
    \Lambda B\Lambda\mathbf 1=0.
\]
Because \(\Lambda\) is invertible, this implies
\[
    B\Lambda\mathbf 1=0.
\]
Thus
\begin{equation}\label{Lambda*1}
    \Lambda\mathbf 1=(\lambda_1,\ldots,\lambda_n)^T\in\ker B.
\end{equation}

We now distinguish the parity of \(n\).

\emph{Case I: \(n\) is odd.}
By Lemma~\ref{lem:null.dm.of.B}, we have $\ker B=\operatorname{span}\{\mathbf 1\}$.
Therefore, by \eqref{Lambda*1}, we must have
\[
    \Lambda\mathbf 1=\lambda\mathbf 1,
\]
or equivalently, $\lambda_1=\cdots=\lambda_n=:\lambda$.

\emph{Case II: \(n\) is even.}
By Lemma~\ref{lem:null.dm.of.B},
\[
    \ker B=\operatorname{span}\{\mathbf 1,v\},
\]
where $v=(1,-1,\ldots,1,-1)^T$.
Hence, by \eqref{Lambda*1}, we must have
\begin{equation}\label{k1*1+k2*v}
    \Lambda\mathbf 1=k_1\mathbf 1+k_2v.
\end{equation}
We claim that \(k_2=0\). Suppose otherwise that \(k_2\neq0\). 
Substituting \eqref{k1*1+k2*v}
into the first equilibrium equation of \eqref{eq.of.n+1-gon}
yields a system whose first two components give
\[
\begin{cases}
m^*[a(1-k_1^2)-k_1k_2a-k_2(k_1+k_2)\widehat a]
=
e_0(k_1+k_2)-m_0,\\[0.6ex]
m^*[a(1-k_1^2)+k_1k_2a-k_2(k_1-k_2)\widehat a]
=
e_0(k_1-k_2)-m_0.
\end{cases}
\]
where $a$ and $\widehat{a}$ are given by \eqref{a.ahat}.
Equivalently,
\[
\begin{cases}
k_1k_2a+k_2^2\widehat a=-k_2e_0,\\[0.6ex]
m^*[a(1-k_1^2)-k_1k_2\widehat a]=k_1e_0-m_0.
\end{cases}
\]
From the first equation we get $k_2\widehat a=-k_1a-e_0$.
Substituting this into the second equation yields
\[
m^*a=-m_0,
\]
which is impossible because $m^*a=m^*\sum_{i=1}^n a_i>0$ and $m_0>0$.
Therefore \(k_2=0\), and hence
\[
    \Lambda\mathbf 1=k_1\mathbf 1.
\]
Thus, also in the even case, $\lambda_1=\cdots=\lambda_n=:\lambda$.

It remains to substitute this common value into the first equilibrium equation of \eqref{eq.of.n+1-gon}.
Since
\[
    A\mathbf 1=\left(\sum_{i=1}^n a_i\right)\mathbf 1,
\]
we obtain
\[
    (A-\Lambda A\Lambda)m^*\mathbf 1
    =
    m^*(1-\lambda^2)A\mathbf 1
    =
    m^*(1-\lambda^2)
    \left(\sum_{i=1}^n a_i\right)\mathbf 1.
\]
On the other hand, the right-hand side is
\[
    (e_0\lambda-m_0)\mathbf 1.
\]
Therefore the equilibrium condition is exactly \eqref{existence.condition}.

Conversely, if $\lambda_1=\cdots=\lambda_n=\lambda$
and the scalar condition above holds, then
\[
    (B-\Lambda B\Lambda)\mathbf 1
    =
    (1-\lambda^2)B\mathbf 1=0,
\]
and
\[
    (A-\Lambda A\Lambda)m^*\mathbf 1
    =
    m^*(1-\lambda^2)A\mathbf 1
    =
    (e_0\lambda-m_0)\mathbf 1.
\]
Hence both equilibrium equations are satisfied. This proves the theorem.
\end{proof}

\textbf{Example:} Consider a regular hexagon (i.e., \( n = 6 \)) with the following conditions:
    \begin{itemize}
        \item The electric charges on the 6 particles at the vertices are \( e_1 = e_2 = \cdots = e_6 = 2 \) (in arbitrary units).
        \item The masses of the particles are \( m_1 = m_2 = \cdots = m_6 = 1 \).
        \item The charge of the central particle is \( e_0 = -\frac{\sqrt{3}}{2} \), and its mass is \( m_0 = \frac{15}{4} \).
    \end{itemize}
    In this system, the electric and gravitational forces will balance each other in such a way that the system can reach equilibrium.

\section{Appendix}

{\bf Lemma 4.4} (Geometric characterization of four-body equilibrium).
{\it 
Let \(A,B,C,D\) be a non-collinear static equilibrium of the charged four-body
problem. With the notation introduced in Section \ref{subsec:4.2}, the following proportional relation
holds:
\[
\begin{aligned}
&\frac{AB}{R_A R_B}:
\frac{BC}{R_B R_C}:
\frac{CD}{R_C R_D}:
\frac{DA}{R_D R_A}:
\frac{AC}{R_A R_C}:
\frac{BD}{R_B R_D}  \\
&=
\Delta_{12}^{2/3}\Delta_{34}^{-1/3}:
\Delta_{23}^{2/3}\Delta_{14}^{-1/3}:
\Delta_{34}^{2/3}\Delta_{12}^{-1/3}:
\Delta_{14}^{2/3}\Delta_{23}^{-1/3}:
\Delta_{13}^{2/3}\Delta_{24}^{-1/3}:
\Delta_{24}^{2/3}\Delta_{13}^{-1/3}.
\end{aligned}
\]
}
\begin{proof}
For \(i\neq j\), set
\[
    r_{ij}=|a_i-a_j|,
    \qquad
    \omega_{ij}=\frac{\delta_{ij}}{r_{ij}^3}.
\]
The equilibrium equations are
\[
    \sum_{j\neq i}\omega_{ij}(a_j-a_i)=0,
    \qquad i=A,B,C,D.
\]
Thus \((\omega_{ij})\) is a self-stress of the complete graph on the four
vertices \(A,B,C,D\).

Since the four points are non-collinear and no three of them are collinear, the
space of affine dependencies among \(A,B,C,D\) is one-dimensional. Let
\((\eta_A,\eta_B,\eta_C,\eta_D)\) be a nonzero affine dependence, namely
\[
    \eta_A+\eta_B+\eta_C+\eta_D=0,
\]
and
\[
    \eta_A A+\eta_B B+\eta_C C+\eta_D D=0.
\]
Then every self-stress of the complete graph \(K_4\) is of the form
\[
    \omega_{ij}=\tau\,\eta_i\eta_j
\]
for some constant \(\tau\neq 0\). Therefore
\[
    |\delta_{ij}|
    =
    |\omega_{ij}|\,r_{ij}^3
    =
    |\tau|\,|\eta_i\eta_j|\,r_{ij}^3.
\]

Moreover, the coefficients \(|\eta_i|\) are proportional to the areas of the
opposite triangles. Denote
\[
    S_A=S_{\Delta BCD},\qquad
    S_B=S_{\Delta ACD},\qquad
    S_C=S_{\Delta ABD},\qquad
    S_D=S_{\Delta ABC},
\]
where \(S_{\Delta XYZ}\) is the area of triangle \(\Delta XYZ\). Absorbing the proportionality
constant into a new constant \(K>0\), we obtain
\[
    \Delta_{ij}=K S_iS_j r_{ij}^3,
    \qquad
    \Delta_{ij}=|\delta_{ij}|.
\]

Now use the circumradius formula
\[
    S_{\Delta XYZ}=\frac{XY\cdot YZ\cdot ZX}{4R_{XYZ}}.
\]
For example,
\[
    S_A=\frac{BC\cdot CD\cdot BD}{4R_A},
    \qquad
    S_B=\frac{AC\cdot CD\cdot DA}{4R_B}.
\]
Hence
\[
    \Delta_{12}
    =
    K S_A S_B\, AB^3
    =
    K'
    \frac{
    AB^3\cdot BC\cdot CD^2\cdot DA\cdot AC\cdot BD
    }{
    R_A R_B
    },
\]
where \(K'>0\) is a constant independent of the pair \(ij\). Similarly,
\[
    \Delta_{34}
    =
    K'
    \frac{
    CD^3\cdot AB^2\cdot BC\cdot DA\cdot AC\cdot BD
    }{
    R_C R_D
    }.
\]
A direct comparison gives
\[
    \left(\frac{AB}{R_A R_B}\right)^3
    =
    C
    \frac{\Delta_{12}^2}{\Delta_{34}},
\]
where the constant
\[
    C=
    \frac{1}{
    K'\,AB\cdot BC\cdot CD\cdot DA\cdot AC\cdot BD
    \cdot R_A R_B R_C R_D
    }
\]
is the same for all six pairs.

Repeating the same computation for the other five pairs yields
\[
    \left(\frac{BC}{R_B R_C}\right)^3
    =
    C
    \frac{\Delta_{23}^2}{\Delta_{14}},\quad
    \left(\frac{CD}{R_C R_D}\right)^3
    =
    C
    \frac{\Delta_{34}^2}{\Delta_{12}},\quad
    \left(\frac{DA}{R_D R_A}\right)^3
    =
    C
    \frac{\Delta_{14}^2}{\Delta_{23}},
\]
and
\[
    \left(\frac{AC}{R_A R_C}\right)^3
    =
    C
    \frac{\Delta_{13}^2}{\Delta_{24}},\quad
    \left(\frac{BD}{R_B R_D}\right)^3
    =
    C
    \frac{\Delta_{24}^2}{\Delta_{13}}.
\]
Taking positive cube roots gives the desired proportional relation.
\end{proof}
\medskip

\noindent{\bf Lemma 6.3}.
{\it 
Let \(B\) be the skew-symmetric circulant matrix defined in \eqref{eq:1}, namely
the matrix associated with the tangential components of the force equations
for the regular \(n\)-gon. Then \(0\) is an eigenvalue of \(B\), and the
corresponding eigenspace is as follows.
\begin{itemize}
    \item If \(n\) is odd, then
    \[
        \ker B=\operatorname{span}\{\mathbf 1\},
    \]
    where $\mathbf 1=(1,1,\ldots,1)^T$.
    In particular, \(0\) is a simple eigenvalue of \(B\).

    \item If \(n\) is even, then
    \[
        \ker B=\operatorname{span}\{\mathbf 1,\mathbf v\},
    \]
    where
    \[
        \mathbf 1=(1,1,\ldots,1)^T,
        \qquad
        \mathbf v=(1,-1,1,-1,\ldots,1,-1)^T.
    \]
    In particular, the multiplicity of the zero eigenvalue of \(B\) is \(2\).
\end{itemize}
}
\begin{proof} Recall
\[
B = \begin{pmatrix}
b_1 & b_2 & \cdots & b_n \\
b_n & b_1 & \cdots & b_{n-1} \\
\vdots & \vdots & & \vdots \\
b_2 & b_3 & \cdots & b_1
\end{pmatrix},
\]
where \(b_1 = 0\) and $b_k$ are given by \eqref{a_i.b_i} for \(k \geq 2\).

By \eqref{a_i.b_i} and direct computation, for \(k \geq 2\), we have
\begin{equation}
\begin{aligned}
b_k &=\operatorname{Im}\left( \frac{1 - \omega^{k-1}}{\left|1 - \omega^{k-1}\right|^3} \right)\\
&= \operatorname{Im}\left( \frac{1 - \cos\left(\frac{k-1}{n}2\pi\right) - i\sin\left(\frac{k-1}{n}2\pi\right)}{\sqrt{\left(1 - \cos\left(\frac{k-1}{n}2\pi\right)\right)^2 + \left(\sin\left(\frac{k-1}{n}2\pi\right)\right)^2}^3} \right)
\\
&= \frac{-\sin\left(\frac{k-1}{n}2\pi\right)}{\sqrt{2 - 2\cos\left(\frac{k-1}{n}2\pi\right)}^3}\\
&= \frac{-\sin\left(\frac{k-1}{n}2\pi\right)}{8\sin^3\left(\frac{k-1}{n}\pi\right)}
= \frac{-\cos\left(\frac{k-1}{n}\pi\right)}{4\sin^2\left(\frac{k-1}{n}\pi\right)}.
\end{aligned}
\end{equation}
where \(\omega = e^{2\pi i / n}\) is the primitive \(n\)-th root of unity.
We then divide the proof into the cases where
\(n\) is odd and where \(n\) is even.

{\bf Case 1: $n$ is odd.}
All eigenvalues of \(B\) are given by (\(k = 0,1,2,\dots,n-1\)):
\begin{equation}
    \lambda_k = b_1 + b_2\omega^k + b_3\omega^{2k} + \cdots + b_n\omega^{(n-1)k}.
\end{equation}

Grouping symmetric terms:
\[
\lambda_k = \left(b_2\omega^k + b_n\omega^{(n-1)k}\right) + \left(b_3\omega^{2k} + b_{n-1}\omega^{(n-2)k}\right) + \cdots + \left(b_{\frac{n+1}{2}}\omega^{\frac{n-1}{2}k} + b_{\frac{n+3}{2}}\omega^{\frac{n+1}{2}k}\right).
\]
Using \(\omega^{m k} - \omega^{(n-m)k} = 2i\sin\left(\frac{2mk\pi}{n}\right)\), we get:
\begin{equation}
    \lambda_k = 2i\sin\left(\frac{2k\pi}{n}\right)b_2 + 2i\sin\left(\frac{4k\pi}{n}\right)b_3 + \cdots + 2i\sin\left(\frac{(n-1)k\pi}{n}\right)b_{\frac{n+1}{2}}.
\end{equation}

 For $k=0$, we immediatly have  \(\lambda_0 = 0 \). 
 We now show \(\lambda_k \neq 0\) for \(k \neq 0\).
This is equivalent to proving:
\begin{equation}
    \sum_{m=1}^{\frac{n-1}{2}} \sin\left(\frac{m}{n} \cdot 2k\pi\right) \frac{\cos\left(\frac{m}{n}\pi\right)}{\sin^2\left(\frac{m}{n}\pi\right)} \neq 0 \quad (k = 1,2,\dots,n-1).
\end{equation}

Now, let's define
\begin{equation*}
    T_k = \sum_{m=1}^{\frac{n-1}{2}} \sin\left(\frac{m}{n} \cdot 2k\pi\right) \frac{\cos\left(\frac{m}{n}\pi\right)}{\sin^2\left(\frac{m}{n}\pi\right)}.
\end{equation*}
In Lemma~\ref{lem:Tk-nonzero} below, we will prove that, for any odd integer $n$,
\[
T_k\neq0,\quad k=1,2,\dots,\frac{n-1}{2}.
\]
Then by the symmetry relation $T_{n-k}=-T_k$, we further have
\[
T_k\neq0,\quad k=\frac{n+1}{2},\frac{n+3}{2},\dots,n-1.
\]
In conclusion, the matrix $B$ possesses exactly one zero eigenvalue.

{\bf Case 2: $n$ is even.}
The eigenvalues of \(B\) are given by (\(k = 0,1,2,\dots,n-1\)):
\[
\lambda_k = b_1 + b_2\omega^k + b_3\omega^{2k} + \cdots + b_n\omega^{(n-1)k}.
\]
Grouping symmetric terms:
\[
\lambda_k = \left(b_2\omega^k + b_n\omega^{(n-1)k}\right) + \left(b_3\omega^{2k} + b_{n-1}\omega^{(n-2)k}\right) + \cdots + \left(b_{\frac{n}{2}}\omega^{\left(\frac{n}{2}-1\right)k} + b_{\frac{n}{2}+2}\omega^{\left(\frac{n}{2}+1\right)k}\right) + b_{\frac{n}{2}+1}\omega^{\frac{n}{2}k}.
\]
Using \(\omega^{m k} - \omega^{(n-m)k} = 2i\sin\left(\frac{2mk\pi}{n}\right)\), \(\omega^{\frac{n}{2}k} = (-1)^k = \cos(k\pi) + i\sin(k\pi)\)
and $b_i=-b_{n+2-i}$ for $2\le i\le n$, 
we get:
\begin{equation}
\lambda_k = 2i\sin\left(\frac{2k\pi}{n}\right)b_2 + 2i\sin\left(\frac{4k\pi}{n}\right)b_3 + \cdots + 2i\sin\left(\frac{(n-2)k\pi}{n}\right)b_{\frac{n}{2}} + \left(\cos(k\pi) + i\sin(k\pi)\right)b_{\frac{n}{2}+1}.
\end{equation}
Note \(b_{\frac{n}{2}+1} = 0\).

For \(k=0\) or \(k=\frac{n}{2}\), we have \(\lambda_0 = \lambda_{\frac{n}{2}} = 0\).
For \(k\neq 0,\frac n2\), Lemma~\ref{lem:Tk-nonzero} gives \(T_k\neq0\).
Hence the corresponding eigenvalue \(\lambda_k\) is nonzero.
Hence, when \(n\) is even, \(B\) has exactly two zero eigenvalues.
\end{proof}

\begin{lemma}\label{lem:Tk-nonzero}
Let \(n\geq 3\), and set
\[
    h=\left\lfloor\frac{n-1}{2}\right\rfloor,
    \qquad
    \theta_m=\frac{m\pi}{n},\quad 1\leq m\leq h.
\]
For \(0\leq k\leq n-1\), define
\begin{equation}\label{T_k}
    T_k
    =
    \sum_{m=1}^{h}
    \frac{
    \sin(2k\theta_m)\cos\theta_m
    }{
    \sin^2\theta_m
    }.
\end{equation}
Then we have
    \[
        T_{n-k}=-T_k,\qquad \forall 1\le k\le n-1,
    \]
and the following statements hold.

\begin{enumerate}
    \item If \(n\) is odd, then
    \[
        T_k>0,\qquad 1\leq k\leq \frac{n-1}{2}.
    \]
    In particular, \(T_k\neq0\) for \(1\leq k\leq n-1\).

    \item If \(n\) is even, then
    \[
        T_0=T_{n/2}=0,
    \]
    and
    \[
        T_k>0,\qquad 1\leq k\leq \frac n2-1.
    \]
    Consequently,
    \[
        T_k\neq0,\qquad k\neq0,\frac n2 .
    \]
\end{enumerate}
\end{lemma}
\begin{proof}
First, since $\sin(2(n-k)\theta_m)=\sin(2m\pi-2k\theta_m)=-\sin(2k\theta_m)$,
from the definition \eqref{T_k}, we have
\[
    T_{n-k}=-T_k.
\]

We now prove the positivity in the range $1\leq k\leq h$.
For convenience, set \(T_0=0\). For \(1\leq k\leq h-1\), we compute the
second difference of \(T_k\). Using
\[
    \sin((2k+2)\theta)-2\sin(2k\theta)+\sin((2k-2)\theta)
    =
    -4\sin(2k\theta)\sin^2\theta,
\]
we obtain
\[
\begin{aligned}
    T_{k+1}-2T_k+T_{k-1}
    &=
    -4\sum_{m=1}^h \sin(2k\theta_m)\cos\theta_m        \\
    &=
    -2\sum_{m=1}^h
    \left[
        \sin((2k+1)\theta_m)+\sin((2k-1)\theta_m)
    \right].
\end{aligned}
\]
We now distinguish two cases.

If \(n\) is odd, then \(h=(n-1)/2\). By the finite sine-sum formula,
\[
    \sum_{m=1}^h \sin(mx)
    =
    \frac{\sin\left(\frac{hx}{2}\right)
          \sin\left(\frac{(h+1)x}{2}\right)}
         {\sin\left(\frac{x}{2}\right)}.
\]
Taking \(x=r\pi/n\) and \(h=(n-1)/2\), we get
\[
\begin{aligned}
    \sum_{m=1}^h \sin(r\theta_m)
    &=
    \frac{
    \sin\left(\frac{(n-1)r\pi}{4n}\right)
    \sin\left(\frac{(n+1)r\pi}{4n}\right)}
    {\sin\left(\frac{r\pi}{2n}\right)}  \\
    &=
    \frac{
    \frac12\left[
    \cos\left(\frac{r\pi}{2n}\right)
    -
    \cos\left(\frac{r\pi}{2}\right)
    \right]}
    {\sin\left(\frac{r\pi}{2n}\right)}.
\end{aligned}
\]
For every odd integer \(r\) with
\(1\leq r\leq n-2\), \(\cos(r\pi/2)=0\). 
Therefore, since \(0<r\pi/(2n)<\pi/2\), we have
\[
    \sum_{m=1}^h \sin(r\theta_m)
    =
    \frac12\cot\left(\frac{r\pi}{2n}\right)>0,
\]
Hence
\[
    T_{k+1}-2T_k+T_{k-1}<0,
    \qquad 1\leq k\leq h-1.
\]

If \(n\) is even, then \(h=n/2-1\). For every odd integer \(r\), 
as a similar argument in the oddness case of $n$, 
one has
\[
\begin{aligned}
    \sum_{m=1}^h \sin(r\theta_m)
    &=
    \frac12\cot\left(\frac{r\pi}{2n}\right)
    -
    \frac12(-1)^{(r-1)/2}.
\end{aligned}
\]
Therefore, for \(1\leq k\leq h-1\),
\[
\begin{aligned}
&\sum_{m=1}^h
\left[
    \sin((2k+1)\theta_m)+\sin((2k-1)\theta_m)
\right]                                                    \\
&\qquad =
\frac12\cot\left(\frac{(2k+1)\pi}{2n}\right)
+
\frac12\cot\left(\frac{(2k-1)\pi}{2n}\right)>0.
\end{aligned}
\]
Thus again
\[
    T_{k+1}-2T_k+T_{k-1}<0,
    \qquad 1\leq k\leq h-1.
\]
In both cases, the sequence
\[
    T_0,T_1,\ldots,T_h
\]
is strictly concave.

It remains to check the endpoint \(T_h\). If \(n\) is odd, then
\(h=(n-1)/2\), and hence
\[
    \sin(2h\theta_m)
    =
    \sin\left(\frac{(n-1)m\pi}{n}\right)
    =
    (-1)^{m+1}\sin\theta_m.
\]
Therefore
\[
    T_h
    =
    \sum_{m=1}^h (-1)^{m+1}\cot\theta_m.
\]
Since \(\cot\theta\) is strictly decreasing on \((0,\pi/2)\), this alternating
sum is positive. Hence \(T_h>0\).

If \(n\) is even, then \(h=n/2-1\), and
\[
    \sin(2h\theta_m)
    =
    \sin\left(\frac{(n-2)m\pi}{n}\right)
    =
    (-1)^{m+1}\sin(2\theta_m).
\]
Thus
\[
    T_h
    =
    2\sum_{m=1}^h
    (-1)^{m+1}
    \frac{\cos^2\theta_m}{\sin\theta_m}.
\]
The function
\[
    f(\theta)=\frac{\cos^2\theta}{\sin\theta}
\]
is positive and strictly decreasing on \((0,\pi/2)\). Therefore the above
alternating sum is again positive, and so \(T_h>0\).

Since \(T_0=0\), \(T_h>0\), and the sequence \(T_0,T_1,\ldots,T_h\) is
strictly concave, every intermediate term lies strictly above the chord joining
\((0,T_0)\) and \((h,T_h)\). Hence
\[
    T_k>0,\qquad 1\leq k\leq h.
\]

If \(n\) is even, then also
\[
    T_{n/2}=0,
\]
because \(\sin(n\theta_m)=\sin(m\pi)=0\). Therefore, for even \(n\),
\[
    T_k\neq0,\qquad k\neq0,\frac n2.
\]
This completes the proof.
\end{proof}

\end{document}